\documentclass[letterpaper,11pt]{article}

\usepackage{microtype}

\usepackage{fullpage}
\usepackage{amsmath}
\usepackage{amsthm}
\usepackage{amsfonts}
\usepackage{amssymb}
\usepackage[usenames,dvipsnames]{xcolor}
\usepackage[colorlinks=true,linkcolor=black,citecolor=black,pagebackref]{hyperref}
\usepackage[noabbrev,capitalize]{cleveref}
\usepackage{svg}
\usepackage{verbatim}
\usepackage{lipsum}
\usepackage{braket}
\usepackage{enumitem}
\usepackage{qcircuit}
\usepackage{array}
\usepackage{thm-restate}
\usepackage{makecell}
\usepackage{listings}
\newcommand{\N}{\mathbb{N}}
\newcommand{\F}{\mathbb{F}}

\newcommand{\calP}{\mathcal{P}}
\newcommand{\calH}{\mathcal{H}}
\newcommand{\calA}{\mathcal{A}}

\newcommand{\calC}{\mathcal{C}}

\newcommand{\poly}{\text{poly}}
\newcommand{\domain}{D}
\newcommand{\range}{R}
\newcommand{\Loss}{\mathcal{L}}
\newcommand{\bp}{P}

\newtheorem{theorem}{Theorem}[section]
\newtheorem{lemma}[theorem]{Lemma}

\newtheorem{corollary}[theorem]{Corollary}
\newtheorem{conjecture}[theorem]{Conjecture}

\newtheorem{proposition}[theorem]{Proposition}
\theoremstyle{definition}
\newtheorem{definition}[theorem]{Definition}
\theoremstyle{remark}
\newtheorem{remark}[theorem]{Remark}
\newtheorem{claim}[theorem]{Claim}

\newcommand\norm[1]{|#1|}

\newcommand{\floor}[1]{\left\lfloor #1 \right\rfloor}
\newcommand{\ceil}[1]{\left\lceil #1 \right\rceil}

\makeatother
\makeatletter

\begin{document}

\title{Cumulative Memory Lower Bounds for Randomized and Quantum Computation}

\author{
Paul Beame\thanks{Research supported by NSF grant CCF-2006359}\\Computer Science \& Engineering\\University of Washington \and Niels Kornerup\\Computer Science\\University of Texas at Austin
}
\date{\today}

\maketitle
\thispagestyle{empty}

\begin{abstract}
Cumulative memory---the sum of space used per step over the duration of a computation---is a fine-grained measure of time-space complexity that was introduced to analyze cryptographic applications like password hashing.
It is a more accurate cost measure for algorithms that have infrequent spikes in memory usage and are run in environments such as cloud computing that allow dynamic allocation and de-allocation of resources during execution, or when many multiple instances of an algorithm are interleaved in parallel.

We prove the first lower bounds on cumulative memory complexity for both sequential classical computation and quantum circuits.
Moreover, we develop general paradigms for bounding cumulative memory complexity inspired by the standard paradigms for proving time-space tradeoff lower bounds that can only lower bound the maximum space used during an execution. 
The resulting lower bounds
on cumulative memory that we obtain
are just as strong as the best time-space tradeoff lower bounds, which
are very often known to be tight.

Although previous results for pebbling and random oracle models have yielded time-space tradeoff lower bounds larger than the cumulative memory complexity, our
results show that in general computational models such separations cannot follow from known lower bound techniques
and are not true for many functions.

Among many possible applications of our general methods, we show that any classical sorting algorithm with success probability at least $1/\poly(n)$ requires cumulative memory $\tilde \Omega(n^2)$, any classical matrix multiplication algorithm requires cumulative memory $\Omega(n^6/T)$, any quantum sorting circuit requires cumulative memory $\Omega(n^3/T)$, and any quantum circuit that finds $k$ disjoint collisions in a random function requires cumulative memory $\Omega(k^3n/T^2)$.
\end{abstract}
\section{Introduction}

For some problems, algorithms can use additional memory for faster running times or additional time to reduce memory requirements. While there are different kinds of tradeoffs between time and space, the most common complexity metric for such algorithms is the maximum time-space (TS) product. This is appropriate when a machine must allocate an algorithm's maximum space throughout its computation. However, recent technologies like AWS Lambda~\cite{BBB+12} suggest that in the context of cloud computing, space can be allocated to a program only as it is needed. When using such services, analyzing the average memory used per step leads to a more accurate picture than measuring the maximum space.

Cumulative memory (CM), the sum over time of the space used per step of an algorithm, is an alternative notion of time-space complexity that is more fair to algorithms with rare spikes in memory.   Cumulative memory complexity was introduced by Alwen and Serbinenko~\cite{AS15} who devised it as a way to analyze time-space tradeoffs for ``memory hard functions'' like password hashes.  Since then, lower and upper bounds on the CM of problems in structured computational models using the black pebble game have been extensively studied, beginning with the work of \cite{AS15, AB16, RD16, ACP+17, ACK+16, ABP17}.  Structured models via pebble games are natural in the context of the random oracle assumptions that are common in cryptography. By carefully interweaving their memory-intensive steps, authors of these papers devise algorithms for cracking passwords that compute many hashes in parallel using only slightly more space than is necessary to compute a single hash. While such algorithms can use parallelism to amortize costs and circumvent proven single instance TS complexity lower bounds, their cumulative memory only scales linearly with the number of computed hashes. 
 Strong CM results have also been shown for the black-white pebble game and used to
derive related bounds for resolution proof systems~\cite{AdRNV17}.

The ideas used for
these structured models yield
provable separations between CM and TS complexity in pebbling and random oracle models. 
The key question that we consider is whether or not the same applies to general models of computation without cryptographic or black-box assumptions: Are existing time-space tradeoff lower bounds too pessimistic for a world where cumulative memory is more representative of a computation's cost? 


\subsection*{Our Results}
The main answer we provide to this question is negative for both classical and quantum computation:
We give \emph{generic methods} that convert existing paradigms for obtaining time-space tradeoff lower bounds involving worst-case space to new lower bounds that replace the time-space product by
cumulative space, immediately yielding a host of new lower bounds on cumulative memory complexity.
With these methods, we show how to extend virtually all known proofs for time-space tradeoffs to equivalent lower bounds on cumulative memory complexity, implying that there cannot be cumulative memory savings for these problems.
Our results, like those of existing time-space tradeoffs, apply in models in which arbitrary sequential computations may be performed between queries to a read-only input.
Our lower bounds also apply to randomized and quantum algorithms that are allowed to make errors.

\subparagraph*{Classical computation} We first focus on lower bound paradigms that apply to computations of multi-output functions $f:\domain^n\rightarrow \range^m$.  Borodin and Cook~\cite{BC82} introduced
a method for proving time-space tradeoff lower bounds for such functions that takes a property such as
the following: for some $K=K(\range,n)$, constant $\gamma$, and distribution $\mu$ on $\domain^n$:
\begin{description}
\item{(*)}  For any partial assignment $\tau$ of $k\le \gamma m$ output values over $\range$ and any restriction (i.e., partial assignment) $\pi$ of $h=h(k,n)$ coordinates on $\domain^n$,
\begin{displaymath}\Pr_{x\sim \mu}[f(x) \mbox{ is consistent with }\tau\mid x\mbox{ is consistent with }\pi]\le K^{-k}.\end{displaymath}
\end{description}
and derives a lower bound of the following form:
\begin{proposition}[\cite{BC82}]
Assume that Property (*) holds for $f:\domain^n\rightarrow \range^m$ with $\gamma>0$ constant. Then,
$T\ (S+\log_2 T)$ is $\Omega(m\ h(S/\log_2 K,n) \log K)$.
\end{proposition}
In particular, since $S\ge \log_2 n$ is essentially always required, if we have the typical case that $h(k,n)=k^\Delta\ h_1(n)$ for some function $h_1(n)$ then this says that $T\cdot S^{1-\Delta}$ is $\Omega(m\ h_1(n)\ \log^{1-\Delta} K)$ or, equivalently, that
$\max(S,\log n)$ is $\Omega([(m\ h_1(n)/T]^{1/(1-\Delta)}\log K)$.
As a simplified example of our new general paradigm, we prove the following analog for cumulative complexity:
\begin{restatable}{theorem}{generic-poly}\label{thm:generic-poly}
Suppose that Property (*) holds for $f:\domain^n\rightarrow \range^m$ with $h(k,n)=k^\Delta h_1(n)$ and $\gamma>0$ constant.
If $T\log_2 T$ is $o(m\ h_1(n) \log K)$ then
any algorithm computing $f$ requires cumulative memory 
\begin{math}
    \Omega\left(\left[(m\ h_1(n))^{1/(1-\Delta)}\log K\right]/T^{\Delta/(1-\Delta)}\right).
\end{math}
\end{restatable}
We note that this bound corresponds exactly to the bound on the product of time and space from Borodin-Cook method.
The full version of our general theorem for randomized computation~(\cref{thm:full-general-bound-classical}) is inspired by an extension by Abrahamson~\cite{Abr91} of the 
Borodin-Cook paradigm to average case complexity. 

We also show how the paradigms for the best time-space tradeoff lower bounds for single-output Boolean functions, which are based on the
densities of \emph{embedded rectangles} where these functions are constant, can be extended to yield cumulative memory bounds. 

\subparagraph*{Quantum computation} We develop an extension of our general approach that applies to quantum computation as well.   
In this case Property (*) and its extensions that we use for our more general theorem must be replaced by statements about quantum circuits with a small number of queries.
In this case, we first
generalize the quantum time-space tradeoff for sorting proven in \cite{KSdW07}, which requires that the
time order in which output values are produced must correspond to the sorted order, to a matching cumulative memory complexity bound of $\Omega(n^3/T)$ that works for any fixed time-ordering of output
production, yielding a more general lower bound. 
(For example, an algorithm may be able to determine the median output long before it determines the other outputs.)
We then show how an analog of our classical general theorem can be applied to extend to paradigms for quantum time-space tradeoffs to cumulative memory complexity bounds for
other problems. 

A summary of our results for both classical and quantum complexity is given in \cref{table-results}.
\begin{table}[t]
\begin{center}
\begin{tabular}{|l|l|l|l|}
\hline
Problem                       & TS Lower Bound                      & Source                                & Matching CM Bound                       \\ \hline

Ranking, Sorting                       & $\Omega(n^2/\log n)$                & \cite{BC82}                           & \cref{thm-classical-sort-cm-bound}                         \\ \hline
Unique Elements, Sorting               & $\Omega(n^2)$                       & \cite{DBLP:journals/siamcomp/Beame91} & \cref{thm-unique-elements}                     \\ \hline
Matrix-Vector Product ($\mathbb{F}$) & $\Omega(n^2 \log |\mathbb{F}|)$     & \cite{Abr91}                          & \cref{thm-mat-vec}                             \\ \hline
Matrix Multiplication ($\mathbb{F}$)        & $\Omega((n^6 \log |\mathbb{F}|)/T)$ & \cite{Abr91}                          & \cref{thm-mat-mult}                            \\ \hline
Hamming Closeness             & $\Omega(n^{2-o(1)})$               & \cite{DBLP:journals/jacm/BeameSSV03}*  & \cref{thm-ham-close}*                          \\ \hline
Element Distinctness          & $\Omega(n^{2-o(1)})$               & \cite{DBLP:journals/jacm/BeameSSV03}*  & \cref{thm-ed}*                                 \\ \hline
Quantum Sorting               & $\Omega(n^3/T)$                     & \cite{KSdW07}                         & \cref{thm-quantum-sort-lb}                     \\ \hline
Quantum $k$ disjoint collisions & $\Omega(k^3 n / T^2)$               & \cite{HM21}                           & \cref{thm-q-k-disj-col}                        \\ \hline
Quantum Boolean Matrix Mult        & $\Omega(n^5/T)$                     & \cite{KSdW07}                         & \cref{thm-q-bool-mat-mul}*                     \\ \hline
\end{tabular}
\end{center}
\caption{All CM bounds match the TS lower bound when considering RAM computation or quantum circuits. The symbol * indicates that the result requires additional assumptions.
}\label{table-results}
\end{table}

\subsubsection*{Previous work}

\subparagraph*{Memory hard functions and cumulative memory complexity}
Alwen and Serbinenko~\cite{AS15} introduced parallel cumulative (memory) complexity as a metric for analyzing the space footprint required to compute \emph{memory hard functions (MHFs)}, which are functions designed to require large space to compute. Most MHFs are constructed using hashgraphs~\cite{DNW05} of DAGs whose output is a fixed length string and their proofs of security are based on pebbling arguments on these DAGs
while assuming access to truly random hash functions for their complexity bounds~\cite{AS15, BCGS16, RD16, ABP17, ACP+17, BZ17}. (See \cref{sec-cm-ts-sep} for
their use in separating CM and TS complexity.)
Recent constructions do not require random hash functions; however, they still rely on cryptographic assumptions \cite{CT19, ABB21}.   

\subparagraph*{Classical time-space tradeoffs}
While these were originally studied in restricted pebbling models similar to those considered to date for cumulative memory complexity~\cite{DBLP:journals/jcss/Tompa80,DBLP:journals/jcss/BorodinFKLT81}, the gold-standard model for time-space tradeoff analysis is that of unrestricted branching programs, which simultaneously capture time and space for general sequential computation.  
Following the methodology of Borodin and Cook~\cite{BC82}, who proved lower bounds for sorting, many other problems have been analyzed~(e.g., \cite{Yes84,DBLP:journals/siamcomp/Abrahamson87,DBLP:conf/focs/Abrahamson90,DBLP:journals/siamcomp/Beame91,DBLP:journals/tcs/MansourNT93}), including universal hashing and many problems in linear algebra~\cite{Abr91}. (See \cite[Chapter 10]{Sav97} for an
overview.)
A separate methodology for single-output functions, introduced in the
context of restricted branching programs~\cite{DBLP:journals/cc/BorodinRS93,okol93}, was extended to general branching programs in~\cite{DBLP:journals/jcss/BeameJS01}, with further applications to other problems~\cite{DBLP:journals/jcss/Ajtai02} including multi-precision integer multiplication~\cite{DBLP:conf/stoc/SauerhoffW03} and error-correcting codes~\cite{DBLP:journals/ipl/Jukna09} as well as over Boolean input domains~\cite{DBLP:journals/toc/Ajtai05,DBLP:journals/jacm/BeameSSV03}.
Both of these methods involve breaking the program into
blocks to analyze the computation under
natural distributions over the inputs based on what happens at the boundaries between blocks.


\subparagraph*{Quantum time-space tradeoffs}
Similar blocking strategies can be applied to quantum circuits to achieve time-space trade-offs for multi-output functions. In~\cite{KSdW07} the authors use direct product theorems to prove time-space tradeoffs for sorting and Boolean matrix multiplication. They also proved somewhat weaker lower bounds for computing matrix-vector products for fixed matrices $A$; those bounds were extended in~\cite{ASdW09} to systems of linear inequalities.  However, both of these latter results apply to computations where the fixed matrix $A$ defining the problem depends on the space bound and, unlike the case of sorting or Boolean matrix multiplication, do not yield a fixed problem for which the lower bound applies at all space bounds.  More recently~\cite{HM21} extended the recording query technique of Zhandry in \cite{Zha19} to obtain time-space lower bounds for the $k$-collision problem and match the aforementioned result for sorting.



\subsubsection*{Our methods}

At the highest level, we employ part of the same paradigms previously used for time-space tradeoff lower bounds. Namely breaking up the computations into blocks of time and analyzing properties of the branching programs or quantum circuits based on what happens at the boundaries between time blocks.
However, for cumulative memory complexity, those boundaries cannot be at fixed locations in time and their selection needs to depend on the space used in these time steps.

Further, in many cases, the time-space tradeoff lower bound needs to set the lengths of those time blocks in a way that depends on the specific space bound.  When extending the ideas to bound cumulative memory usage, there is no single space bound that can be used throughout the computation; this sets up a tricky interplay between the choices of boundaries between time blocks and the lengths of the time blocks.  Because the space usage within a block may grow and shrink radically, even with optimal selection of block boundaries, the contribution of each time block to the overall cumulative memory may be significantly lower than the time-space product lower bound one would obtain for the individual block.  

We show how to bound any loss in going from time-space tradeoff lower bounds to cumulative memory lower bounds in a way that depends solely on the bound on the lengths of blocks as a function $h_0$ of the target space bound (cf. \cref{lem-loss}).  For many classes of bounding functions we are able to bound the loss by a constant factor, and we are able show that it is always at most an $O(\log n)$ factor loss. 
If this bounding function $h_0$ is non-constant, we also need to bound the optimum way for the algorithm to allocate its space budget for producing the require outputs throughout its computation. 
This optimization again depends on the bounding function $h_0$.
This involves minimizing a convex function based on $h_0$ subject to a mix of convex and concave constraints, which is not generally tractable. 
However, assuming that $h_0$ is nicely behaved, we are able to apply specialized convexity arguments (cf. \cref{lem-concave}) which let us derive strong lower bounds on cumulative memory complexity.

\subparagraph*{Road map}

We give the overall definitions in \cref{sec-prelim}, including a review of the standard definitions of the work space used by quantum circuits.
\cref{sec-sort} is stand-alone section containing a simpler explicit cumulative memory lower bound for classical sorting algorithms that does not rely on our general theorems.
In~\cref{sec-qsort}, we give our lower bound for
quantum sorting algorithms which gives a taste of the issues involved for our general theorems. 
In~\cref{sec-generic}, we give the general theorems that let us convert the Borodin-Cook-Abrahamson paradigm for multi-output functions to cumulative memory lower bounds for classical randomized algorithms; that section also contains the corresponding theorems for quantum lower bounds.
\cref{sec-classical-applications} applies our general theorems from \cref{sec-generic} to lower bound the cumulative memory complexity for some concrete problems. \cref{sec-single} proves our lower bounds for single output functions.

\cref{sec-cm-ts-sep} gives a random oracle separation between the time-space product and cumulative memory.
Some technical lemmas that allow us to generalize lower bounds for quantum sorting to arbitrary success probabilities are in \cref{sec-quant-sort-output-bounds}.
\cref{append-optimization-lemma,sec-loss} contain some of the arguments that bound
the optimum allocations of cumulative space budgets to time steps and allow us to bound the loss functions.


\section{Preliminaries}
\label{sec-prelim}
Cumulative memory is an abstract notion of time-space complexity that can be applied to any model of computation with a natural notion of space. Here we will use branching programs and quantum circuits as concrete models, although our results generalize to any reasonable model of computation.
\subparagraph*{Branching Programs}
 Branching programs with input  $\{x_1, \ldots, x_n\} \in \domain^n$ are known as $\domain$-way branching programs and are defined using a rooted DAG in which each non-sink vertex is labeled with an $i \in [n]$ and has $|\domain|$ outgoing edges that correspond to possible values of $x_i$.
Each edge is optionally labeled by some number of
output statements expressed as pairs
$(j,o_j)$ where $j\in [m]$ is an output index and $o_j\in \range$ (if outputs are to be 
ordered) or simply $o_j\in \range$ (if outputs are to be unordered).
Evaluation starts at the root $v_0$ and follows the appropriate labels of the respective $x_i$. 
We consider branching programs $\bp$ that contain $T+1$ layers where the outgoing edges from nodes in each layer $t$ are all in layer $t+1$.
We impose no restriction on the query pattern of the branching program or when it can produce parts of the output. 
Such a branching program $P$ has the following complexity measures: The \emph{time} of the branching program is $T(\bp) = T$. The \emph{space} of the branching program is $S(\bp) = \max_t \log_2 \norm{L_t}$ where $L_t$ is the set of nodes in layer $t$. 
Observe that in the absence of any limit on its space, a branching program could equally well be a decision tree; hence the minimum time for branching 
programs to compute a function $f$ is its \emph{decision
tree complexity}.
The \emph{time-space} (product) used by the branching program is $TS(\bp) = T(\bp) S(\bp)$. The \emph{cumulative memory} used by the branching program is $CM(\bp) = \sum_t \log_2 \norm{L_t}$.

Branching programs are very general and simultaneously model time and space for sequential computation. 
In particular they model time and space for random-access off-line multitape Turing machines 
and random-access machines (RAMs) when time is unit-cost, space is log-cost, and the input and output are read-only
and write-only respectively\footnote{In prior work, branching program space has often been defined to be the logarithm of the total number of nodes (e.g., \cite{BC82, Abr91}) rather than
the logarithm of the width (maximum number of nodes per layer), though the latter has been used (e.g., \cite{CFL83}).
The natural conversion from an arbitrary space-bounded machine to a branching program produces one that is not 
leveled (i.e., nodes are not segregated by time step).  After leveling the branching program, the space of the original machine becomes the logarithm of the width
 (cf. \cite{Pip79}). 
The width-based definition is also the only natural one by which to measure cumulative memory complexity and, in any case, the two definitions differ by at most the additive $\log_2 T$ we used for the Borodin-Cook bound, with lower bounds on width implying lower bounds on size.}.
Branching programs are much more flexible than these models since they can make arbitrary changes to their storage in a single step.

\subparagraph*{Quantum Circuits} We also consider quantum circuits $\calC$ classical read-only input $X=x_1, \ldots, x_n$ that can be queried using an XOR query oracle.
As is normal in circuit models, each output wire is associated
with a fixed position in the output sequence, independent of the input.
As shown in \cref{fig-quantum-circuit} following~\cite{KSdW07}, we abstract an arbitrary quantum circuit $\calC$ into layers $\calC = \{L_1, \ldots, L_T\}$ where layer $L_t$ starts with the $t$-th query $Q$ to the input and ends with the start of the next layer.
During each layer, an arbitrary unitary transformation $V$ gets applied which can express an arbitrary sub-circuit involving input-independent computation.
The sub-circuit/transformation $V$ outputs $S_t$ qubits for use in the next layer in addition to some qubits that are immediately measured in the standard basis, some of which are treated as classical write-only output.
The time of $\calC$ is lower bounded by the number of layers $T$ and we say that the space of layer $L_t$ is $S_t$.
Observe that to compute a function $f$, $T$ must be at least the
\emph{quantum query complexity} of $f$ since that measure corresponds
the above circuit model when the space is unbounded.
Note that the cumulative memory of a circuit is lower-bounded by the sum of the $S_t$. For convenience we define $S_0$, the space of the circuit before its first query, to be zero. Thus we only consider the space after the input is queried.


\begin{figure}[htb]
\centering
\includegraphics[width=0.8\textwidth]{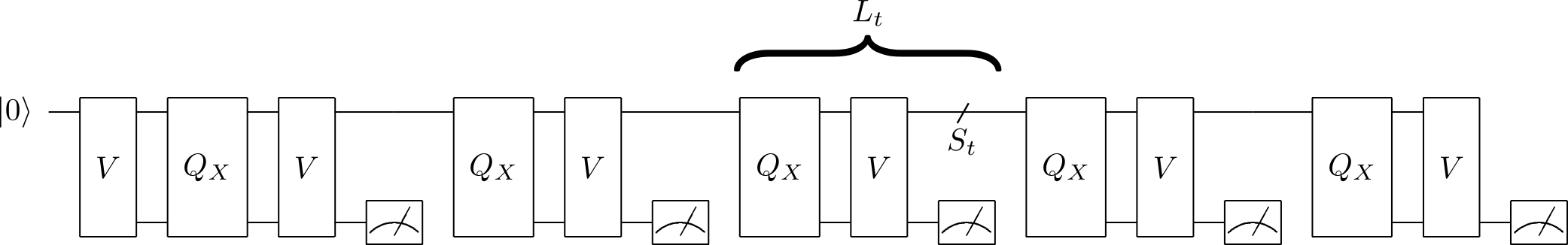}
\caption{The abstraction of a quantum circuit into layers.}\label{fig-quantum-circuit}
\end{figure}

\section{Cumulative memory complexity of classical sorting algorithms}
\label{sec-sort}

For a natural number $N$, the standard version of \emph{sorting} is a function $Sort_{n,N}:[N]^n \rightarrow [N]^n$ that on input $x\in [N]^n$ produces an output $y\in [N]^n$ in non-decreasing order where $y$ is a permutation of $x$; that is, there is some permutation $\pi$ such that $y_i=x_{\pi(i)}$ for all $i\in [n]$.   A related problem is the \emph{ranking} problem $Rank_{n,N}:[N]^n\rightarrow [n]^n$ which on input $x\in [N]^n$ produces
a permutation $\pi$ represented
as the vector $(\pi(1),\ldots,\pi(n))$ such that 
$Sort_{n,N}(x)=(x_{\pi(1)},\ldots,x_{\pi(n)})$
and whenever $x_i = x_j$ for
$i<j$ we have $\pi(i)<\pi(j)$.

\begin{proposition}[\cite{BC82}]
\label{prop-rank-sort}
\begin{itemize}
\item[(a)] If there is an $[nN]$-way branching program $P$ computing $Sort_{n,nN}$ then there is an $[N]$-way branching program $P'$ computing $Rank_{n,N}$ with
$T(P')\le T(P)$, 
$S(P')\le S(P)$, and
$CM(P')\le CM(P)$.
\item[(b)] If there is an $[N]$-way branching program $P''$
computing $Rank_{n,N}$  then there is an $[N]$-way branching program $P'''$ computing $Sort_{n,N}$ with
$T(P''')\le 2T(P'')$, 
$S(P''')\le S(P'')+\log_2 N$, and
$CM(P''')\le 2CM(P'')+T(P''')\log_2 N$.
\end{itemize}
\end{proposition}

\begin{proof}
    For part (a), the program $P'$ is exactly $P$ except that when $P$ queries $x_i\in [Nn]$,
    $P'$ reads $x'_i\in [N]$
    and branches on value
    $x_i=(x'_i,i)$ and
    when $P$ outputs 
    $(i,y_i)=(i,x_{\pi(i)})$ on an edge for $x_{\pi(i)}=(x'_{\pi(i)},\pi(i))$, $P'$
    outputs $(i,\pi(i))$.
    For part (b), the
    program $P'''$ is exactly
    $P''$ except that whenever $P''$ outputs
    $(i,\pi(i))$ on an edge, $P'''$ queries $x_{\pi(i)}$ and outputs $(i,x_{\pi(i)})$.
    One layer becomes two layers and the number of nodes per layer of
    $P'''$ is at most $N$
    times that of $P''$.
\end{proof}

Following~\cite{BC82}, we focus on inputs where the $x_i$ are distinct.
In this case, the tie-breaking we
enforced in defining $Rank_{n,N}$ when there are equal elements is irrelevant.

\begin{proposition}[\cite{BC82}]\label{lem-sort-cond-c}
There is an $\alpha>0$ such that
the following holds. 
Let $n$ be sufficiently large and $\mu$ be the uniform distribution over lists of $n$ distinct integers from $[n^2]$. Then for any branching program $B$ of height $h\leq \alpha n$ and for all integers $k \leq 2\alpha n$, the probability for $x \sim \mu$ that $B$ produces at least $k$ correct output values of $Rank_{n,n^2}$ on input $x$ is at most $2^{-k/\ceil{\log_2 n}}$.
\end{proposition}


\begin{theorem}\label{thm-classical-sort-cm-bound}
Let $P$ be a branching program computing
$Sort_{n,n^3}$ with
probability at least
$n^{-O(1)}$ and $T=T(P)$. Then $T $ is $\Omega(n^2/\log^2 n)$ or $CM(P)$ is $\Omega(n^2 / \log n)$. Further, any random access machine computing
$Sort_{n,n^3}$ with $n^{-O(1)}$ probability requires cumulative memory of $\Omega(n^2 / \log n)$ bits.
\end{theorem}

\begin{proof}
We prove the same bounds
for branching programs $P$ computing $Rank_{n,n^2}$ which, by~\cref{prop-rank-sort}, implies the bounds for computing
$Sort_{n,n^3}$.

For simplicity we first assume that $P$ is determistic and is always correct. Let $\alpha$ be the constant 
and $\mu$ be the probability distributuon on $[n^2]^n$ from \cref{lem-sort-cond-c},
and let $H=\floor{\frac{\alpha}{2} n}$. 
We partition $P$ into $\ell = \ceil{T/H}$ intervals $\{I_1, \ldots, I_{\ell}\}$, all of length $H$ except for the first, which may be shorter than the rest. 
Let $t_1 = 0$, $t_{\ell+1} = T$, and for $i \in [2,\ell]$, $t_i$ be the time-step in $I_{i}$ with the fewest number of nodes. We define $S_i = \log_2(\norm{L_{t_i}})$ where $L_j$ is the set of nodes of $P$ in layer $j$. The $i$-th time block $B_i$ will contain all layers from $t_i$ to $t_{i+1}$. We observe: 
\begin{equation}\label{eqn-cmc-sort}
CM(P) \geq \sum_{i=2}^{\ell} S_i \,H=H\ \sum_{i=1}^\ell S_i
\end{equation}
since $S_1=0$.
Define $k_i = \ceil{\ceil{\log_2 n } (S_i + \log_2 (2T))}$, which will be our target number of outputs for block $B_i$. 
By our choice of $B_i$ we know its length is at most $\alpha n$ and it starts at a layer with $2^{S_i}$ nodes. 
So, by \cref{lem-sort-cond-c}, combined with a union bound, the probability for $x\sim \mu$ that $B_i$ produces at least $k_i$ correct output values of $Rank_{n,n^2}$ on input $x \sim \mu$ is at most $1/(2T)$. 
Thus the probability over $\mu$ that at least one block $B_i$ produces at least $k_i$ correct output values is at most $1/2$ and the probability that the total number of outputs produced is at most $\sum_{i=1}^{\ell} (k_i -1)$ is at least $1/2$. 
Since $P$ must always produce $n$ correct outputs, we must have:
\begin{displaymath}\sum_{i=1}^\ell (k_i - 1) \geq n.\end{displaymath}
Inserting the definition of $k_i$ we get:
\begin{displaymath}\sum_{i=1}^\ell (\ceil{\log_2 n} (S_i + \log_2 (2T))) \geq n.\end{displaymath}
Using \cref{eqn-cmc-sort} to express this in terms of $CM(P)$ gives us:
\begin{displaymath}CM(P)/H + \ell\, \log_2(2T) \geq \frac{n}{\ceil{\log_2 n}}\end{displaymath}
or
\begin{displaymath}CM(P) + T \log_2(2T) \geq \frac{n\floor{\frac{\alpha}{2} n}}{\ceil{\log_2 n}} \geq \frac{\alpha n^2}{3 \log_2 n}.\end{displaymath}
Thus at least one of
$T \log_2 (2T)$  
or
$CM(P)$ is at least $\alpha n^2/(6 \log_2 n)$, as required, since $\log T$
is $O(\log n)$ wlog.
The bound for random-access machines comes from observing that such a machine requires at least one memory cell of $\Omega(\log T)$ bits at every time step.

To prove the bound for algorithms with success probability $n^{-c}$, we multiply
$\log_2(2T)$ in the above argument by $(c+1)$.   
Since any sorting algorithm must have
$T\ge n$, on randomly chosen inputs the probability that it produces at least
$\sum_{i=1}^{\ell} (k_i -1)$ correct outputs becomes
$\frac{1}{2n^c}<\frac{1}{n^c}$ and hence the above
bounds (reduced by the constant factor $c+1$) apply to deterministic
algorithms with success probability $1/n^c$ for inputs from the uniform distribution over lists of $n$ distinct integers from $[n^2]$.  
By Yao's lemma this implies the same lower bound for 
randomized algorithms with success probability $n^{-c}$.
\end{proof}

\cref{thm-classical-sort-cm-bound} applies to cumulative working memory of any algorithm that produces its sorted output in a write-only output
vector and can compute those values in
arbitrary time order.   
If the algorithm is constrained to produce its sorted
output in the natural time order then, following~\cite{DBLP:journals/siamcomp/Beame91}, one can obtain a slightly
stronger bound.

\begin{theorem}\label{thm-sort-fixed-order}
    Any branching program $P$
    computing the outputs of $Sort_{n,n}$ in order in time $T$ and probability at least $4/5$ requires
    $T$ to be $\Omega(n^2/\log n)$ or
    $CM(P)$ to be $\Omega(n^2)$. 
    Further,
    any random access machine computng $Sort_{n,n}$ in order with probability at least $4/5$
    requires cumulative memory
    $\Omega(n^2)$.
\end{theorem}

\begin{proof}[Proof Sketch]
    Any such algorithm can easily determine all
    the elements of the input that
    occur uniquely and the lower bounds follow from the bounds
    on Unique Elements that we prove in~\cref{sec-classical-applications}.
\end{proof}
\section{Quantum cumulative memory complexity of sorting}
\label{sec-qsort}

As an illustrative example, we 
first show 
that the quantum cumulative memory complexity of sorting is $\Omega(n^3/T)$, matching the $TS$ complexity bounds given in \cite{KSdW07, HM21}. 
This involves the quantum circuit model which, as we have noted, produces each output position at a predetermined input-independent layer. We restrict our attention to circuits that output all elements in the input in some fixed rank order. While our proof is inspired by the time-space lower bound of \cite{KSdW07}, it can be easily adapted to follow the proof in \cite{HM21} instead.
We start by constructing a probabilistic reduction from the $k$-threshold problem to sorting.
\begin{definition}
In the \emph{$k$-threshold problem} we receive an input $X = x_1, \ldots, x_n$ where $x_i \in \{0,1\}$. We want to accept iff there are at least $k$ distinct values for $i$ where $x_i = 1$.
\end{definition}


\begin{proposition}[Theorem 13 in \cite{KSdW07}]\label{thm-q-kthresh-dp}
For every $\gamma > 0$ there is an $\alpha > 0$ such that any quantum $k$-threshold circuit with at most $T \leq \alpha \sqrt{kn}$ queries and with perfect soundness must have completeness $\sigma \leq e^{-\gamma k}$ on inputs with Hamming weight $k$.
\end{proposition}

\begin{lemma}\label{lem-qcircuit-fixed-outputs-bound}
Let $\gamma > 0$. Let $n$ be sufficiently large and $\calC(X)$ be a quantum circuit with input $X = x_1, \ldots, x_n$. There is a  $\beta< 1$ depending only on $\gamma$ such that for all $k \leq \beta^2 n$ and $R \subseteq \{n/2+1, \ldots, n\}$ where $\norm{R} = k$, if $\calC(X)$ makes at most $\beta \sqrt{kn}$ queries, then the probability that $\calC(X)$ can correctly output all $k$ pairs $(x_i, r_j)$ where $r_j \in R$ and $x_i$ is the $r_j$-th smallest element of $X$ is at most $e^{(1-\gamma) k-1}$. If $R$ is a contiguous set of integers, then the probability is at most $e^{- \gamma k}$.
\end{lemma}

A version of 
this lemma
was first proved in \cite{KSdW07} with the additional assumption that the set of output ranks $R$ is a contiguous set of integers;  this was sufficient to show that any quantum circuit that produces its sorted output in sorted time order requires that $T^2S$ is $\Omega(n^3)$. 
The authors stated that their proof can be generalized to any fixed rank ordering, but the generalization is not obvious. We generalize their lemma to non-contiguous $R$, which is sufficient to obtain an $\Omega(n^3/T)$ lower bound on the cumulative complexity of sorting independent of the
time order in which the sorted output is produced.

\begin{proof}[Proof of \cref{lem-qcircuit-fixed-outputs-bound}]
Choose $\alpha$ as the constant for $\gamma$ in \cref{thm-q-kthresh-dp} and let $\beta = \sqrt{2} \alpha / 6$. 
Let $\calC$ be a circuit with at most $\beta \sqrt{kn}$ layers that outputs the $k$ correct pairs $(x_i, r_j)$ with probability $p$. 
Let $R = \{r_1, \ldots r_k\}$ where $r_1 <r_2 < \ldots < r_k$.
We describe our construction of a circuit $\calC'(X)$ solving the $k$-threshold problem on inputs $X = x_1, \ldots, x_{n/2}$ with exactly $k$ ones in terms of 
a function $f: [n/2] \to R$.
Given $f$, we re-interpret the input as follows: we replace each $x_i$ with $x_i' = f(i)x_i$, 
add $k$ dummy values of $0$, and add one dummy value of $j$ for each $j \in \{n/2+1, \ldots, n\} \setminus R$. 
Doing this gives us an input $X' = x_1', \ldots, x_n'$ that has $n/2$ zeroes. 
If we assume that $f$ is 1-1 on the $k$ ones of $X$,
then the image of the ones of $X$ will
be $R$ and there will be precisely one element of $X'$
for each $j\in \{n/2+1,\ldots,n\}$. 
Therefore the
element of rank $j>n/2$ in $X'$ will have value $j$, and hence the rank $r_1, \ldots, r_k$
elements of $X'$ will be the images of precisely those elements of $X$ with $x_i=1$.

To obtain perfect soundness, we cannot rely on the output of $\calC(X')$ and must be able to check that
each of the output ranks was truly mapped to by a distinct one of $X$.
For each element $x_i$ of $X$ we simply append its index $i$ as $\log_2 n$ low order bits to its image $x_i'$ and append an all-zero bit-vector of length $\log_2 n$ to each dummy value to obtain input $X''$.
Doing so will not change the ranks of the elements in $X'$, but will allow recovery of the $k$ indices that should be the ones in $X$.
In particular, circuit $\calC'(X)$ will run $\calC(X'')$ and then for each output $x''_j$ with low order bits $i$, $\calC'(X)$ will query $x_i$, accepting if and only if all of those $x_i=1$.
More precisely, since the mapping from each $x_i$ to the corresponding $x_i''$ is only a function of $f$, $x_i$, and $i$, as long as $\calC'(X)$ has an explicit representation of $f$, it can simulate each query of $\calC(X'')$ with two oracle queries to $X$. 
Since $\calC'$ has at most
\begin{displaymath}2\beta \sqrt{kn} + k \leq 3\beta \sqrt{kn} \leq \alpha \sqrt{kn/2}\end{displaymath}
layers, by \cref{thm-q-kthresh-dp}, it can only accept with probability $\leq e^{-\gamma k}$ on inputs with $k$ ones.

We now observe that for each fixed $X$ with exactly $k$ ones, for a randomly chosen function
$f:[n/2]\to R$, the probability that $f$ is 1-1 on the
ones of $X'$ is exactly $k!/k^k\ge e^{1-k}$. Therefore $\calC'(X)$ will give the indices of the $k$ ones in $X$ with probability\footnote{Note that though this is exponentially small in $k$ it
is still sufficiently large compared to the 
completeness required in the lower bound for the $k$-threshold problem.} at least $p\cdot e^{1-k}$.
However, this probability must be at most $e^{-\gamma k}$, so we can conclude that $p \leq e^{(1-\gamma) k-1}$. In the event that $R$ is a contiguous set of integers, observe that any choice for the function $f$ will make $X''$ have the ones of $X$ become ranks $r_1, \ldots, r_k$. So the probability of finding the ones is at least $p \leq e^{-\gamma k}$.
\end{proof}

By setting $k$ and $\gamma$ appropriately, \cref{lem-qcircuit-fixed-outputs-bound} gives a useful upper bound on the number of fixed ranks successfully output by any $\beta \sqrt{Sn}$ query quantum circuit that has access to $S$ qubits of input dependent initial state. To handle input-dependent initial state, we will need to use the following proposition.

\begin{restatable}[\cite{Aar05}]{proposition}{aaronson}\label{prop-quant-union}
Let $\calC$ be a quantum circuit, $\rho$ be any $S$ qubit (possibly mixed) state, and $I$ be the $S$ qubit maximally mixed state. If $\calC$ with initial state $\rho$ produces some output $\mathcal{O}$ with probability $p$, then $\calC$ with initial state $I$ produces $\mathcal{O}$ with probability at least $p/2^{2S}$.
\end{restatable}

This allows us to bound the overall progress made by any short quantum circuit.

\begin{restatable}{lemma}{qsortout}\label{cor-quantum-sort-outputs}
There is a constant $\beta > 0$ such that, for any fixed set of $S\le\beta^2 n$ ranks that are greater
than $n/2$, the probability that any quantum circuit $\calC$ with at most $\beta \sqrt{Sn}$ queries and $S$ qubits of input-dependent initial state correctly produces 
the outputs for these $S$ ranks is at most $1/e$.
\end{restatable}

\begin{proof}
Choose $\beta$ as the constant when $\gamma$ is $1+\ln(4)$ in \cref{lem-qcircuit-fixed-outputs-bound}.
 Applying \cref{prop-quant-union} to the bound in \cref{lem-qcircuit-fixed-outputs-bound} gives us that a quantum circuit with $S$ qubits of input-dependent state can produce a fixed set of  $k\le \beta^2 n$ outputs larger than median with a probability at most $2^{2S} e^{(1-\gamma)k-1}$. 
Since $\gamma = 1+\ln(4)$ setting $k=S$ gives that this probability is $\leq 1/e$.
\end{proof}


\begin{theorem}\label{thm-quantum-sort-lb}
When $n$ is sufficiently large, any quantum circuit $\calC$ for sorting a list of length $n$ with success probability at least $1/e$ and at most $T$ layers that produces its sorted outputs in any fixed time order requires cumulative memory that is $\Omega(n^3/T)$.
\end{theorem}

\begin{proof}
We partition $\calC$ into blocks with large cumulative memory that can only produce a small number of outputs. We achieve this by starting at last unpartitioned layer and finding a suitably low space layer before it so that we can apply \cref{cor-quantum-sort-outputs} to upper bound the number of correct outputs that can be produced in that block with a success probability of at least $1/e$. Let $\beta$ be the constant from \cref{cor-quantum-sort-outputs} and $k^*(t)$ be the least non-negative integer value of $k$ such that the interval:
\begin{displaymath}I(k,t) = \left[t-\frac{\beta}{2} (2^{k+1}-1)\sqrt{n}, t - \frac{\beta}{2}(2^{k}-1)\sqrt{n}\right]\end{displaymath}
contains some $t'$ such that $S_{t'} \leq 4^k - 1$. We recursively define our blocks as follows. Let $\ell$ be the number of blocks generated by this method. The final block $\calC_\ell$ starts with the first layer $t_{\ell-1} \in I(k^*(T),T)$ where $S_{t_{\ell-1}} \leq 4^{k^*(T)}-1$ and ends with layer $t_{\ell} = T$. Let $t_{i}$ be the first layer of block $\calC_{i+1}$. Then the block $\calC_{i}$ starts with the first layer $t_{i-1} \in I(k^*(t_{i}),t_{i})$ where $S_{t_{i-1}} \leq 4^{k^*(t_{i})} - 1$ and ends with $t_{i}$. See \cref{fig-quantum-segments} for an illustration of our partitioning. Since $S_0 = 0$ we know that $k^*(t) \leq \log(T)$. Likewise since $S_t > 0$ when $t > 0$, for all $t > \frac{\beta}{2} \sqrt{n}$ we know that $0 < k^*(t) \leq \log(T)$.

\begin{figure}[t]
\centering
\includegraphics[scale=.8]{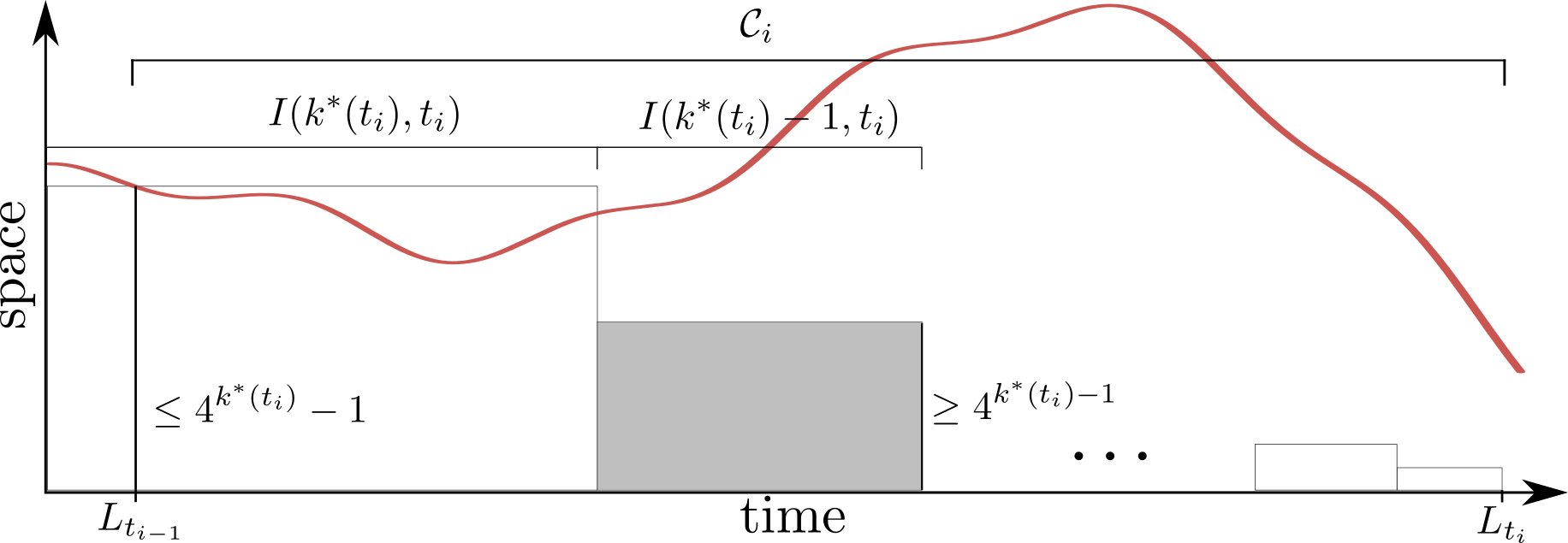}
\caption{How we define the block $\calC_i$ that ends at layer $L_{t_i}$. The red line is a plot of $\calC$'s space over time. The grey layers are the ones used to lower bound the cumulative memory complexity of $\calC_i$, as each of these layers uses at least $4^{k^*(t_i) - 1}$ qubits and the length of this interval is $\frac{\beta}{2} 2^{k^*(t_i)-1} \sqrt{n}$. }\label{fig-quantum-segments}
\end{figure}

Block $\calC_i$ starts with less than $4^{k^*(t_{i})}$ qubits of initial state and has length at most $\beta 2^{k^*(t_i)} \sqrt{ n}$; so by \cref{cor-quantum-sort-outputs}, if
$4^{k^*(t_i)}\le \beta^2 n$, the block $\calC_i$ can output at most $4^{k^*(t_i)}$ inputs with failure probability at most $1/e$. Additionally $\calC_i$ has at least $\frac{\beta}{2} 2^{k^*(t_i) -1}\sqrt{n}$ layers so
\begin{equation}
    \sum_{i=1}^\ell \frac{\beta}{4}2^{k^*(t_i)}\sqrt{n} \le T
\end{equation}
and each of these layers has at least $4^{k^*(t_i) - 1}$ qubits\footnote{This may not hold for $\calC_1$ with length less than $\frac{\beta}{2} \sqrt{N}$, but \cref{lem-qcircuit-fixed-outputs-bound} gives us that this number of layers is insufficient to find a fixed rank input with probability at least $1/e$. Thus we can omit such a block from our analysis.}, so
the cumulative memory of $\calC_i$ is at least $\frac{\beta}{2} 2^{3k^*(t_i) - 3}\sqrt{n}$ so
\begin{equation}
    CM(\calC) \geq \sum_{i=1}^\ell \frac{\beta}{2} 2^{3k^*(t_i) -3} \sqrt{n}.
\end{equation}
We now have two possibilities:  If we have some $i$ such that
$4^{k^*(t_i)}> \beta^2 n$, the cumulative memory of
$\calC_i$ alone is at least $\beta^4 n^2/16$ which is $\Omega(n^2)$
and hence $\calC$ has cumulatively memory $\Omega(n^3/T)$ since $T\ge n$.
Otherwise, since we require that the algorithm is correct with
probability at least $1/e$, each block $\calC_i$ can produce at most
$4^{k^*(t_i)}$ outputs.
Since our circuit must output all $n/2$ elements larger than the median, we know $\sum_{i=1}^\ell 4^{k^*(t_i)} \geq n/2$. For convenience we define $w_i = 2^{k^*(t_i)}$ which allows us to express the constraints as
\begin{equation}
\label{eqn-quant-cm-bound}
CM(\calC) \geq  \frac{\beta}{16} \sqrt{n} \sum_{i=1}^\ell w_i^{3}
\hfil\textrm{ and }\hfil 
    \frac{\beta}{4}\sqrt{n}\sum_{i=1}^\ell  w_i\le T
    \hfil\textrm{ and }\hfil
    \sum_{i=1}^\ell w_i^2 \geq n/2.
\end{equation}
 Minimizing $\sum_{i=1}^\ell w_i^3$ is a non-convex optimization problem and
 can instead be solved using
 \begin{equation}
     \textrm{Minimize}\quad \sum_{i=1}^\ell x_i^3
     \hfil\textrm{ subject to }\hfil
     \sum_{i=1}^\ell x_i^2\ge \xi
     \hfil\textrm{ and }
     \sum_{i=1}^\ell x_i\le \xi
     \hfil\textrm{ and }\hfil\forall i, x_i\ge 0,
 \end{equation}
 for 
  $x_i = \displaystyle{\frac{ 8 T}{ \beta n^{3/2}} w_i}$ and $\xi = \displaystyle{\frac{32T^2}{\beta^2 n^2}}$.
\cref{lem-moments} from Appendix C shows that for non-negative $x_i$ with $\sum x_i \leq \sum x_i^2$, we have $\sum x_i^2 \leq \sum x_i^3$.
Thus \begin{math}\sum x_i^3 \geq \xi\end{math} and applying the variable substitution gives us:
\begin{math}\displaystyle{\sum_{i=1}^\ell w_i^{3} \geq \frac{\beta n^{5/2}}{16T}}.\end{math}
Plugging this into \cref{eqn-quant-cm-bound} gives us the bound:
\begin{math}\displaystyle{ CM(\calC) \geq \frac{\beta^2 n^3}{256 T}}\end{math}
and hence the cumulative memory of $\calC$ is $\Omega(n^3/T)$.
\end{proof}

In \cref{sec-quant-sort-output-bounds} we also show how we can change the length of the blocks to generalize the above proof to arbitrary success probabilities.
\section{General methods for proving cumulative memory lower bounds}\label{sec-generic}

Our method involves adapting techniques previously used to prove tradeoff lower bounds on worst-case time and worst-case space. 
We show that the same properties that yield lower bounds on the product of time and space in the worst case can also be used to produce
nearly identical lower bounds on cumulative memory.
To do so, we first revisit the standard approach to
such time-space tradeoff lower bounds.

\subsection*{The standard method for time-space tradeoff lower bounds for multi-output functions}

Consider a multi-output function $f$ on $\domain^n$
where the output $f(x)$ is either unordered (the output is simply a set of
elements from $\range$) or ordered (the output is a vector of elements from $\range$).  
Then $|f(x)|$ is either the size of the set 
or the length of the vector of elements.
The standard method for obtaining an ordinary time-space tradeoff lower bounds for multi-output functions  on $\domain$-way branching programs is the following:

\subparagraph*{The part that depends on $f$:}

Choose a suitable probability distribution $\mu$ on $\domain^n$, often simply the uniform distribution on $\domain^n$ and then:
\begin{description}
\item[(A)] Prove that  $\Pr_{x\sim \mu}[|f(x)|\ge m]\ge \alpha$.
\item[(B)] Prove that for all $k\le m'$ and
any branching program $B$ of height $\le h'(k,n)$,
the probability for $x\sim \mu$ that $B$ produces at least $k$ correct output values of $f$ on input $x$ is at most $C\cdot K^{-k}$ for some $m'$, $h'$, $K=K(R,n)$, and constant $C$ independent of $n$.
\end{description}
Observe that under any distribution $\mu$, a branching program with ordered outputs that makes no queries can produce $k$ outputs that are all correct with probability at least $|\range|^{-k}$, so the bound in (B) shows that, roughly, up to the difference between $K$ and $|\range|$ there is not much gained by using a branching program of
height $h$.  

\subparagraph*{The generic completion:}
In the following outline we omit integer rounding for readability.
\begin{itemize}
\item Let $S'=S+\log_2 T$ and suppose that 
\begin{equation}
\label{eqn-big-S}
S'\le   m'\log_2 K - \log_2(2C/\alpha).
\end{equation}
\item 
Let $k=[S'+\log_2 (2C/\alpha)]/\log_2 K$, which is at most $m'$ by hypothesis on $S'$, and define $h(S',n)=h'(k,n)$.
\item Divide time $T$ into $\ell=T/h$ blocks of length $h=h(S',n)$.  
\item 
The original branching program can be split into at most $T\cdot 2^S=2^{S'}$ sub-branching programs of height $\le h$,
each beginning at a boundary node between layers. 
By Property (B) and a union bound, for $x\sim \mu$ the probability that at least one of these $\le 2^{S'}$ sub-branching programs of height at most $h$ produces $k$ correct outputs on input $x$ is at most
\begin{math}2^{S'}\cdot C\cdot K^{-k}\le \alpha/2\end{math}
by our choice of $k$.
\item Under distribution $\mu$, by (A), with probability at least $\alpha$,
an input $x\sim\mu$ has some block of time where at least $m/\ell=m\cdot h(S',n)/T$ outputs of $f$ must be produced on input $x$.
\item If $m \cdot h(S',n)/T\le k$, this can occur for at most an $\alpha/2$ fraction of inputs under $\mu$.
Therefore we have \begin{math}m \cdot h(S',n)/T> k=[S'+\log_2 (2C/\alpha)]/\log_2 K\end{math} and hence since
$h(S',n)\ge h(S,n)$, combining with
\cref{eqn-big-S}, we have
\begin{displaymath}T\cdot (S+\log_2 T)=T\cdot S' \ge \min \left(m\  h(S,n),\ m'\  n'\right)\  \log_2 K-\log_2(C/\alpha)\cdot T\end{displaymath}
where $n'\le n$ is the decision tree complexity of $f$ and
hence a lower bound on $T$.
\end{itemize}

\begin{remark}
Though it will not impact our argument, for many instances of the above outline, the proof of Property (B) is shown for
a decision tree of the same height by proving an analog for the conditional probability along each path in the decision tree separately; this will apply to the tree as a whole since the paths are followed by disjoint inputs, so
Property (B) follows from the alternative property below:
\begin{description}
\item[(B')]  For any partial assignment $\tau$ of $k\le m'$ output values over $\range$ and any restriction (i.e., partial assignment) $\pi$ of $h'(k,n)$ coordinates within $\domain^n$,
\begin{displaymath}\Pr_{x\sim \mu}[f(x) \mbox{ is consistent with }\tau\mid x\mbox{ is consistent with }\pi]\le C\cdot K^{-k}.\end{displaymath}
\end{description}
Observe that Property (B') is only a slightly more general version of Property (*) from the introduction where $C=1$, $m'$ is arbitrary, and $h'$ is used instead of $h$.
\end{remark}

\begin{remark}
The above method still gives lower bounds for many multi-output functions $g:\domain^N\rightarrow \range^M$ that have individual output values that are easy to compute or large portions of the input space on which they are easy to compute.   The bounds follow by applying the method to some
subfunction $f$ of $g$ given by $f(x)=\Pi_O(g(x,\pi))$ where $\pi$ is 
a partial assignment to the input coordinates and 
$\Pi_O$ is a projection onto a subset $O$ of output coordinates.
In the subsequent discussions we ignore this issue, but the idea can
be applied to all of our lower bound methods.
\end{remark}

\subsection*{A general extension to cumulative memory bounds}

To give a feel for the basic ideas of the method, we first 
show this for a simple case.
Observe that, other than the separate bound on time, the lower bound on cumulative memory usage we prove in this case is asymptotically identical to the bound achieved for the product of time and worst-case space using the standard outline. 

\begin{theorem}\label{thm-simple-general-bound-classical}
Let $c>0$.  Suppose that properties (A) and (B) apply for $h'(k,n)=h(n)$, $m'=m$, and $\alpha=C=1$.
If
\begin{math}T \log_2 T \le \frac{m\  h(n)\  \log_2 K}{6 (c+1) }\end{math}
then the cumulative memory used in computing $f:\domain^n\rightarrow\range^m$ in time $T$ with success
probability at least $T^{-c}$ is at least \begin{math}\frac{1}{6}\, m\  h(n)   \log_2 K.\end{math}
\end{theorem}

\begin{proof}
Fix a deterministic branching program $P$ of length $T$ computing $f$.
Rather than choosing fixed blocks of height $h=h(n)$, layers of nodes at a fixed distance from each other, and a fixed target of $k$ outputs per block, we choose the block boundaries depending on the properties of $P$ and the target $k$ depending on the property of
the boundary layer chosen.

Let $H=\lfloor h(n)/2\rfloor$.  
We break $P$ into $\ell=\lceil T/H\rceil $ time segments of length $H$ working backwards from step $T$ so that the first
segment may be shorter than the rest.    
We let $t_1=0$ and for $1< i\le \ell$ we let 
\begin{math}t_{i}=\arg\min\{\ |L_t|\ :\ T-(\ell-i+1)\cdot H\le t < T-(\ell-i)\cdot H\ \}\end{math}
be the time step with the fewest nodes among all time steps $t\in [T-(\ell-i+1)\cdot H,T-(\ell-i)\cdot H]$.

\begin{figure}[t]
    \centering
    \includegraphics[width=.8\textwidth]{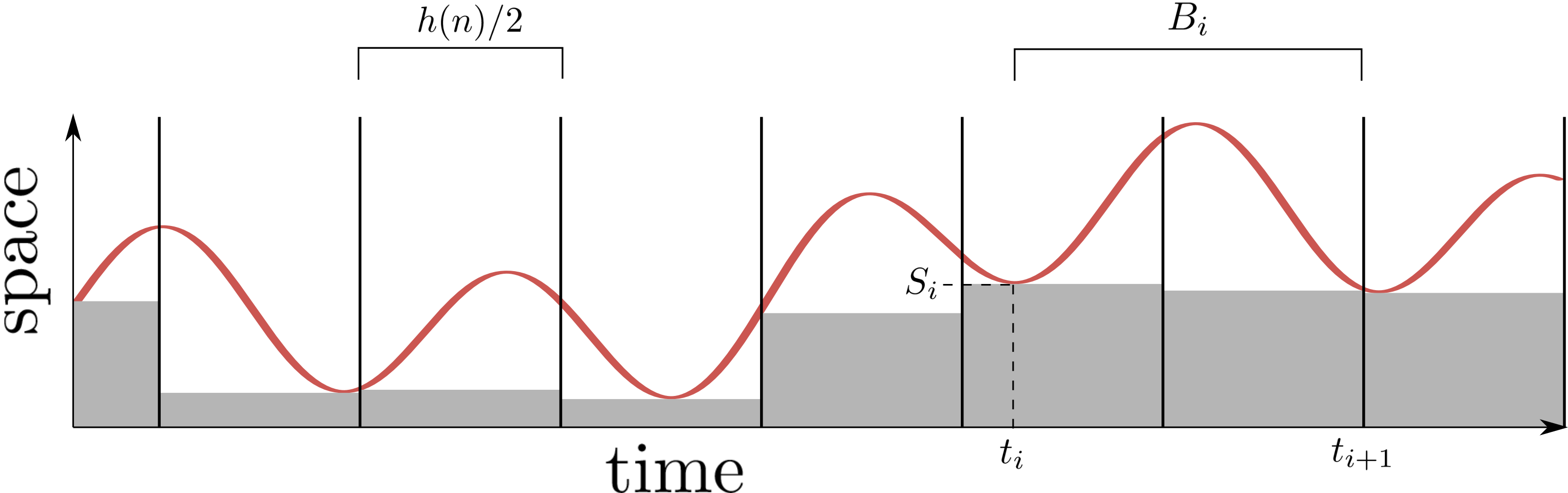}
    \caption{Our generic method for choosing blocks when $h(k,n) = h(n)$.
    The area marked in grey corresponds to the cumulative memory lower bound we obtain.
    }
    \label{fig-simple-generic-construction}
\end{figure}

The $i$-th time block of $P$ will be between times $t_i$ and $t_{i+1}$.
Observe that by construction $|t_{i+1}-t_i|\le h(n)$ so each block has length at most $h(n)$.
This construction is shown in \cref{fig-simple-generic-construction}.
Set $S_i=\log_2 |L_{t_i}|$ so that $L_{t_i}$ has at $2^{S_i}$ nodes.
By definition of each $t_i$, the cumulative memory used by
$P$, 
\begin{equation}
CM(P)\ge \sum_{i=1}^{\ell} S_i\cdot H.\label{cmc-formula}
\end{equation}  
(Note that since $S_1=0$, it does not matter that the first segment is shorter than the rest\footnote{This simplifies some calculations and is the prime reason for starting the time segment boundaries at $T$ rather than at 0.}.)

We now define the target $k_i$ for the number of output values produced in each time block to be the smallest integer such that
$K^{-k_i}\le 2^{-S_i}/T^{c+1}$.
That is,
\begin{displaymath}k_i=\lceil (S_i +(c+1)\log_2 T)/\log_2 K\rceil.\end{displaymath}
For $x\sim \mu$, for each $i\in [\ell]$ and each sub-branching program $B$ rooted at some
node in $L_{t_i}$ and extending until time $t_{i+1}$, by our choice
of $k_i$ and
Property (B), if $k_i\le m$, the probability that
$B$ produces at least $k_i$ correct outputs on input $x$ is at most
$2^{-S_i}/T^{c+1}$.
Therefore, by a union bound, for $x\sim \mu$ the probability that
$P$ produces at least $k_i$ correct outputs in the $i$-th time
block on input $x$ is at most
\begin{math}|L_{t_i}| \cdot 2^{-S_i}/T^{c+1}=1/T^{c+1}.\end{math}
Therefore, if each $k_i\le m$, the probability for $x\sim \mu$ that there is some $i$ such that $P$ produces at least $k_i$ correct outputs on input $x$
during the $i$-th block is at most $\ell/T^{c+1}< T^c$.   
Therefore, if each $k_i\le m$, the probability for $x\sim \mu$ that $P$ produces at most $\sum_{i=1}^\ell (k_i-1)$ correct outputs in total on input $x$ is $>1-1/T^c$.

If each $k_i\le m$, since $P$ must produce $m$ correct outputs on $x\in \domain^n$ with probability at least $1/T^c$, we must have $\sum_{i=1}^\ell (k_i-1)\ge m$. 
On the other hand, if some $k_i>m$ we have the same bound.
Using our definition of $k_i$ we have
\begin{math}\sum_{i=1}^\ell [ (S_i +(c+1)\log_2 T)]/\log_2 K]\ge m\end{math}
or 
\begin{math}\sum_{i=1}^\ell (S_i +(c+1)\log_2 T)\ge m\cdot \log_2 K.\end{math}
Plugging in the bound (\ref{cmc-formula}) on the cumulative
memory and the value of $\ell$, it implies that
\begin{math}CM(P)/H + (c+1) \lceil T/H\rceil\cdot \log_2 T\ge m\cdot \log_2 K\end{math}
or that
\begin{math}CM(P) + (c+1) T\log_2 T\ge \frac{1}{3}\  m\cdot h(n)\cdot \log_2 K,\end{math}
where the 3 on the right rather than a 2 allows us to remove
the ceiling. 
Therefore either
\begin{displaymath}T\log_2 T> \frac{m\cdot h(n)\cdot \log_2 K}{6(c+1)}\label{large-T}\hfil\textrm{ or }\hfil
CM(P)\ge \frac{1}{6}\  m\  h(n)\  \log_2 K.\qedhere\end{displaymath}
\end{proof}
 
In the general version of our theorem there
are a number of additional complications, most especially because the branching program
height limit $h(k,n)$ in Property (B) \emph{can depend on}  $k$, the target for the number of
outputs produced. 
This forces the lengths of the blocks and the
space used at the boundaries 
between blocks to depend on each other in a
quite delicate way.
In order to discuss the impact of that dependence and state our general theorem,
we need the following definition.

\begin{definition}
Given a non-decreasing function $p:\mathbb{R}\rightarrow\mathbb{R}$ with $p(1)=1$, we define $p^{-1}:\mathbb{R}\rightarrow\mathbb{R}\cup \{\infty\}$ by $p^{-1}(R)=\min\{j\mid p(j)\ge k\}$.
We also define the \emph{loss}, $\Loss_p$, of $p$ by
\begin{displaymath}\Loss_p(n)= \min_{1\le k\le p(n) } \frac{\sum_{j=1}^k p^{-1}(j)}{k\cdot p^{-1}(k)}.\end{displaymath}
\end{definition}

\begin{restatable}{lemma}{losslemma}
\label{lem-loss}
The following hold for every non-decreasing function $p:\mathbb{R}\rightarrow\mathbb{R}$ with $p(1)=1$:
\begin{enumerate}[label=(\alph*)]
\item $1/p(n)\le \Loss_p(n)\le 1$. 
\item If $p$ is a polynomial function $p(s)=s^{1/c}$
then $\Loss_p(n)>1/2^{c+1}$.
\item For any $c>1$, $\Loss_p(n)\ge\displaystyle \min_{1\le s\le n}\frac{p(s)-p(s/c)}{cp(s)}$.
\item We say that $p$ is \emph{nice} if it is differentiable and there is an integer $c>1$ such that for all $x$, $p'(cx)\ge p'(x)/c$.  If $p$ is nice then $\Loss_p(n)$ is $\Omega(1/\log_2 n)$.
This is tight for $p$ with $p(s)=1+\log_2 s$.
\end{enumerate}
\end{restatable}

We prove these technical statements in \cref{sec-loss}.    
Here is our full general theorem.

\begin{theorem}\label{thm:full-general-bound-classical}
Let $c>0$.
Suppose that function $f$ defined on $\domain^n$ has properties (A) and (B) with $\alpha$ that is $1/n^{O(1)}$ and
$m'$ that is $\omega(\log_2 n)$.
For $s>0$, define $h(s,n)$ to be $h'(k,n)$ for $k=s/\log_2 K$. 
Suppose that
$h(s,n)= h_0(s)\, h_1(n)$ with $h_0(1)=1$ and $h_0$ is constant or a differentiable function such that
$s/h_0(s)$ is increasing and concave.
Define $S^*=S^*(T,n)$ by
\begin{displaymath}\frac{S^*}{h_0(S^*)}=\frac{m\  h_1(n)\   \log_2 K}{6T}.\end{displaymath}
\begin{enumerate}[label=(\alph*)]
\item
Either
\begin{math}\displaystyle{T \log_2 (2CT^{c+1}/\alpha) > \frac{1}{6}\, m\  h_1(n)\  \log_2 K},
\end{math}
which implies that $T$ is $\Omega(\frac{m \  h_1(n)\  \log K}{\log n}),$
or the cumulative memory used by a
randomized branching program in computing $f$ in time $T$ with error $\varepsilon\le \alpha(1-1/(2T^c))$ is at least \begin{displaymath}\frac{1}{6}\,\Loss_{h_0}(n\log_2 |\domain|)\cdot \min\left(m\  h(S^*(T,n),n),\ 3m'\  h'(m'/2,n)\right) \cdot \log_2 K.\end{displaymath}
\item Further any randomized
random-access machine computing $f$ in time $T$ with error
$\varepsilon \le \alpha(1-1/(2T^c))$
requires cumulative memory
\begin{displaymath}\Omega\left(\Loss_{h_0}(n\log_2 |\domain|)\cdot 
\min\left(m\ h(S^*(T,n),n),\ m'\ h'(m'/2,n)\right)\cdot \log_2 K\right).\end{displaymath}
\end{enumerate}
\end{theorem}

Before we give the proof of the theorem, 
we note that by \cref{lem-loss}, in the case that $h_0$ is constant or $h_0(s)=s^\Delta$ for some constant $\Delta>0$, 
which together account for all existing applications we are aware of, the function $\Loss_{h_0}$ is
lower bounded by a constant.
In the latter case, $h_0$ is differentiable, has $h_0(s)=1$, and the function $s/h_0(s)=s^{1-\Delta}$ is increasing and concave so it satisfies the conditions of our theorem.  By using $\alpha=1$, $m'=m$, and $C=1$ with $h$ from Property (*) in place of $h'$ in Property (B'), \cref{thm:full-general-bound-classical} yields \cref{thm:generic-poly}.

More generally, the value $S^*$ in the statement
of this theorem is at least a constant factor times the value of $S$
used in the generic time-space tradeoff lower bound methodology.
Therefore, for example, the cumulative memory lower bound
derived for random-access machines via \cref{thm:full-general-bound-classical} is close to the 
lower bound on the product of time and worst-case space given by standard methods.

\begin{proof}[Proof of \cref{thm:full-general-bound-classical}]
We prove both (a) and (b) directly for
branching programs, which can model random-access machines, and will describe the small variation that occurs in the case that the
branching program in question comes from a random-access machine.
To prove these properties for randomized branching programs,
by Yao's Lemma~\cite{DBLP:conf/focs/Yao77} it suffices to prove
the properties for deterministic branching programs that have error at most $\varepsilon$
under distribution $\mu$.
Fix a (deterministic) branching program $P$ of length $T$ computing $f$ with error at most
$\varepsilon$ under distribution $\mu$. 
Without loss of generality, $P$ has maximum space usage at most $S^{max}=n\log_2 |\domain|$ space since there
are at most $|\domain^n|$ inputs.

Let $H=\lfloor h_1(n)/2\rfloor$.  
We break $P$ into $\ell=\lceil T/H\rceil $ time segments of length $H$ working backwards from step $T$ so that the first
segment may be shorter than the rest.    
We then choose a sequence of \emph{candidates} for the time steps in which to begin new blocks, as follows:
We let $\tau_1=0$ and for $1< i\le \ell$ we let 
\begin{displaymath}\tau_{i}=\arg\min\{\ |L_t|\ :\ T-(\ell-i+1)\cdot H\le t < T-(\ell-i)\cdot H\ \}\end{displaymath}
be the time step with the fewest nodes among all time steps $t\in [T-(\ell-i+1)\cdot H,T-(\ell-i)\cdot H]$.
Set $\sigma_i=\log_2 |L_{\tau_i}|$ so that $L_{\tau_i}$ has at $2^{\sigma_i}$ nodes.
This segment contributes at least $\sigma_i\cdot H$ to the cumulative memory bound of $P$.

To choose the beginning $t_{i^*}$ of the last time block\footnote{Since we are working backwards from
the end of the branching program and we do not know
how many segments are included in each block, we don't actually know this index until things stop with $t_1=0$}.
we find the smallest $k$ such that $h_0(\sigma_{\ell-k+1})< k$.
Such a $k$ must exist since $h_0$ is a non-decreasing non-negative function, $h_0(1)=1$ and $\sigma_1=0< 1$. 
We now observe that the length of the last block 
is at most $k\cdot H$ which by choice of $k$ is less than
$h(\sigma_{\ell-k+1},n)$ and hence we have 
satisfied the requirements for Property (B) to apply at
each starting node of the last time block.

By our choice of each $\tau_i$, the cumulative memory used in the last $k$ segments is at least
\begin{math}\sum_{j=1}^{k}\sigma_{\ell+1-j} \cdot H.\end{math}
Further, since $k$ was chosen as smallest with the above
property,
we know that for every $j\in [k-1]$ we have \begin{math}h_0(\sigma_{\ell-j+1})\ge j\end{math}
Hence we have $\sigma_{\ell-j+1}\ge h_0^{-1}(j)$ and we
get a cumulative memory bound for the last $k$ segments of at least
\begin{equation}
(\sigma_{\ell-k+1}+\sum_{j=1}^{k-1} h_0^{-1}(j))\cdot H.\label{cm-lb-per-segment}
\end{equation}

\begin{claim}
$\sigma_{\ell-k+1}+\sum_{j=1}^{k-1} h_0^{-1}(j)\ge \Loss_{h_0}(S^{max})\cdot \sigma_{\ell-k+1}\cdot k$.
\end{claim}

\begin{proof}[Proof of Claim]
Observe that it suffices to prove the claim when we
replace
$\sigma_{\ell-k+1}$, which appears on both sides, by a larger quantity.
In particular, we show how to prove the claim
with $h_0^{-1}(k)$ instead, which is larger since
$h_0(\sigma_{\ell-k+1})<k$.
But this follows immediately
since by definition \begin{math}\Loss_{h_0}(S^{max})\le \frac{\sum_{j=1}^{k} h_0^{-1}(j)}{k\cdot h_0^{-1}(k)},\end{math}
which is equivalent to what we want to prove.
\end{proof}

Write $S_{i^*}=\sigma_{\ell-k+1}$. 
By the claim, the cumulative memory contribution associated with the
last block beginning at $t_{i^*}$ is at least 
\begin{math}\Loss_{h_0}(S^{max})\cdot S_{i^*}\cdot h_0(S_{i^*}) H.\end{math}

We repeat this in turn to find the time step for the beginning of the next
block from the end, $t_{i^*-1}$.  
One small difference now is that there is a last partial
segment of height at most $H$ from the beginning of segment containing $t_{i^*}$ to layer $t_{i^*}$.
However, this only adds at most $h_1(n)/2$ to the 
length of the segment which still remains well within the height bound of $h(S_{i^*-1},n)=h_0(S_{i^*-1})h_1(n)$ for Property (B) to apply.    

Repeating this back to the beginning of the branching program we obtain a decomposition of the branching
program into some number $i^*$ of blocks, the $i$-th block beginning at
time step $t_i$ with $2^{S_i}$ nodes, height between $h_0(S_i)H$ and 
$h_0(S_i)H+H\le 2h_0(S_i) H$, and with an associated cumulative
memory contribution in the $i$-th block of \begin{math}\ge \Loss_{h_0}(S^{max})\cdot S_{i}\cdot h_0(S_{i}) H.\end{math}
(This is correct even for the partial block starting
at time $t_1=0$ since $S_1=0$.)
Since we know that $i^*\le \ell$, for convenience, we also define $S_i=0$ for $i^*+1\le i\le \ell$.
Then, by definition
\begin{align}
&CM(P)\ge \Loss_{h_0}(S^{max})\cdot \left (\sum_{i=1}^{i^*}S_i\cdot h_0(S_i)\right)\cdot H=\Loss_{h_0}(S^{max})\cdot \left (\sum_{i=1}^{\ell}S_i\cdot h_0(S_i)\right) \label{cm:general}
\\
\label{general:length}
&\qquad\textrm{ and }\qquad\sum_{i=1}^\ell h_0(S_i) \le T/H.
\end{align}

As in the previous argument for the simple case, for $i\le i^*$, we define the target $k_i$ for the number of output values produced in each time block to be the smallest integer such that
$C\dot K^{-k_i}\le 2^{-S_i}\alpha/(2T^{c+1})$.
That is,
\begin{math}k_i=\lceil (S_i + \log_2(2CT^{c+1}/\alpha))/\log_2 K\rceil.\end{math}

If $k_i>m'$ for some $i$, then $S_i\ge m'\cdot \log_2 K - \log_2 (2CT^{c+1}/\alpha)\ge (m' \log_2 K)/2$ since $m'$ is $\omega(\log n)$ and $1/\alpha$ and $T$ are $n^{O(1)}$.
Therefore
$h_0(S_i)\ge h'(m'/2,n)$ and hence
\begin{displaymath}CM(P)\ge \frac{1}{2}\, \Loss_{h_0}(S^{max})\cdot m'\cdot h'(m'/2,n)\cdot \log_2  K\end{displaymath}

Suppose instead that $k_i\le m'$ for all $i\le i^*$.
Then, for $x\sim \mu$, for each $i\in [i^*]$ and each sub-branching program $B$ rooted at some
node in $L_{t_i}$ and extending until time $t_{i+1}$, by our choice
of $k_i$ and
Property (B), the probability that
$B$ produces at least $k_i$ correct outputs on input $x$ is at most
\begin{math}\alpha \cdot 2^{-S_i}/(2T^{c+1}).\end{math}
Therefore, by a union bound, for $x\sim \mu$ the probability that
$P$ produces at least $k_i$ correct outputs in the $i$-th time
block on input $x$ is at most
\begin{displaymath}|L_{t_i}| \cdot \alpha\cdot 2^{-S_i}/(2T^{c+1})=\alpha/(2T^{c+1})\end{displaymath}
and hence the probability for $x\sim \mu$ that there is some $i$ such that $P$ produces at least $k_i$ correct outputs on input $x$
during the $i$-th block is at most $\ell\cdot \alpha/(2T^{c+1})< \alpha/(2T^c)$.   
Therefore, the probability for $x\sim \mu$ that $P$ produces at most $\sum_{i=1}^\ell (k_i-1)$ correct outputs in total on input $x$ is 
$>1-\alpha/(2T^c)$.

Since, by Property (A) and the maximum error it allows, $P$ must produce at least $m$ correct outputs with probability at least
$\alpha-\epsilon\ge \alpha-\alpha(1-1/(2T^c))=\alpha/(2T^c)$  for $x\sim \mu$, we must have $\sum_{i=1}^{i^*} (k_i-1)\ge m$. 
Using our definition of $k_i$ we obtain
\begin{displaymath}\sum_{i=1}^{i^*} (S_i +\log_2 (2CT^{c+1}/\alpha))\ge m\,\log_2 K.\end{displaymath}

This is the one place in the proof where there is a distinction between an arbitrary 
branching program and one that comes from 
a random access machine.

We first start with the case of arbitrary
branching programs:
Note that $i^*\le \ell=\lceil T/H\rceil=\lceil T/\lfloor h_1(n)/2\rfloor\rceil$.
Suppose that $T \log_2 (2CT^{c+1}/\alpha) \le \frac{1}{6}\, m \cdot h_1(n)\cdot \log_2 K$.
Then,
even with rounding, we obtain
$\sum_{i=1}^{i^*} S_i\ge \frac{1}{2}\,m\, \log_2 K.$ 

\begin{sloppypar}
Unlike an arbitrary branching program that may do non-trivial computation with sub-logarithmic $S_i$, a random-access 
machine with even one register requires at least $\log_2 n$ bits of memory (just to index
the input for example) and hence $S_i+\log_2 (2CT^{c+1}/\alpha)$ will be $O(S_i)$, since $T$ is at most polynomial in $n$ without loss of generality and $1/\alpha$ is at most polynomial in $n$ by assumption.
Therefore we obtain that
$\sum_{i=1}^{i^*} S_i$ is $\Omega(m\, \log_2 K)$ without the
assumption on $T$.
\end{sloppypar}

In the remainder we continue the argument
for the case of arbitrary branching programs
and track the constants involved.   
The same argument obviously applies for programs coming from random-access machines with slightly different constants that we will not track.  In particular, since $S_i=0$ for $i>i^*$ we have
\begin{equation}
\sum_{i=1}^{\ell} S_i\ge \frac{1}{2}\, m\cdot \log_2 K.  \label{general:space-constraint}
\end{equation}

From this point we need to do something different from
the argument in the simple case because the lower bound on the total cumulative memory contribution is given by
\cref{cm:general} and is not simply $\sum_{i=1}^{\ell} S_i\cdot H$.
Instead, we combine \cref{general:space-constraint} and \cref{general:length} using the following technical lemma that we prove in \cref{append-optimization-lemma}.

\begin{restatable}{lemma}{concavelemma}
\label{lem-concave}
Let 
$p:\mathbb{R}^{\ge 0}\rightarrow \mathbb{R}^{\ge 0}$ be a differentiable function such that $q(x)=x/p(x)$ is a concave increasing function of $x$.
For $x_1,x_2,\ldots\in \mathbb{R}^{\ge 0}$, if
$\sum_i x_i \ge K$ and $\sum_i p(x_i)\le L$  then $\sum_i x_i p(x_i) \ge q^{-1}(K/L)\cdot L$.
\end{restatable}

In our application of \cref{lem-concave},
$p=h_0$, $K=\frac{1}{2}\, m\cdot \log_2 K$, and $L=T/H$.
Let $S^*$ be the solution to
\begin{math}\frac{S^*}{h_0(S^*)}=K/L=\frac{m\cdot H\cdot  \log_2 K}{2T}\ge \frac{m\cdot h_1(n) \log_2 K}{6T}.\end{math}
Then 
\cref{lem-concave} implies that
\begin{math}\sum_{i=1}^{\ell} S_i\cdot h_0(S_i)\ge S^*\cdot T/H=\frac{1}{2}, m\cdot h_0(S^*) \cdot  \log_2 K.\end{math}
and hence
\begin{displaymath}CM(P)\ge\Loss_{h_0}(S^{max})\cdot\frac{1}{2}\, m\cdot h_0(S^*) \cdot H\cdot  \log_2 K
\ge \frac{1}{6}\, \Loss_{h_0}(S^{max})\cdot m\cdot h(S^*,n)\cdot \log_2 K
\end{displaymath}
since $H=\lfloor h_1(n)/2\rfloor$ and
$h(S^*,n)=h_0(S^*)\cdot h_1(n)$. 
\end{proof}
 
In the special case that $h_0(s)=s^\Delta$ (and indeed for any nice function $h_0$), there is an alternative variant of the above 
in which one breaks up time into exponentially growing
segments starting with time step $T$.
We used that alternative approach in \cref{sec-qsort}.

\begin{remark}\label{rem-generic-space-bounded}
If we restrict our attention to $o(m' \log K)$-space bounded computation, then each $k_i \leq m'$ and the cumulative memory bound for a branching program in \cref{thm:full-general-bound-classical} becomes
\begin{math}\frac{1}{6}\, \Loss_{h_0}(n\log_2 |\domain|)\cdot m\cdot h(S^*(T,n),n) \cdot \log_2 K.\end{math}
And the bound for RAM cumulative memory becomes
\begin{math}\Omega\left(\Loss_{h_0}(n\log_2 |\domain|)\cdot 
m\cdot h(S^*(T,n),n)\cdot\log_2 K\right).\end{math}
\end{remark}
\subsection*{Generic method for quantum time-space tradeoffs}\label{sec-q-generic}
Quantum circuit time-space lower bounds have the same general structure as their classical branching program counterparts. They require a lemma similar to (B) that gives an exponentially small probability of producing $k$ outputs with a small number of queries.

\begin{lemma}[Quantum generic property]\label{lem-quant-generic-lemma}
For all $k \leq m'$ and any quantum circuit $\calC$ with at most $h'(k,n)$ layers, there exists a distribution $\mu$ such that when $x \sim \mu$, the probability that $\calC$ produces at least $k$ correct output values of $f(x)$ is at most $C \cdot K^{-k}$.
\end{lemma}

Such lemmas have historically been proving using direct product theorems \cite{KSdW07, ASdW09} or the recording query technique \cite{HM21}. Quantum time-space tradeoffs use the same blocking strategy as branching programs; however, they cannot use union bounds to account for input dependent state at the start of a block. Instead, \cref{prop-quant-union} lets us apply \cref{lem-quant-generic-lemma} to blocks in the middle of a quantum circuit.

The $2^{2S}$ factor in \cref{prop-quant-union} means that a quantum time-space or cumulative memory lower bound will be half of what you would expect from a classical bound with the same parameters. Since a quantum circuit must have $\log_2 n$ qubits to make a query, we know that the space between layers is always at least $\log_2 n $. Therefore the generic time-space tradeoff for quantum circuits is
\begin{displaymath}
T\cdot S \text{ is } \Omega \left(\min\{m \  h'(S,n), m' \  Q(f)\} \cdot  \log_2 K\right)
\end{displaymath}
where $Q(f)$ is the bounded-error quantum query complexity of $f$.

\subsection*{Generic method for quantum cumulative complexity bounds}
Our generic argument can just as easily be applied to quantum lower bounds for problems where we have an instance of \cref{lem-quant-generic-lemma} using \cref{prop-quant-union} to bound the number of outputs produced even with initial input-dependent state.
Since quantum circuits require at least $\log_2 n$ qubits to hold the query
index, the bounds derived are like those from \cref{thm:full-general-bound-classical}(b).
\begin{corollary} \label{cor-quant-generic}
Let $c>0$. Suppose that function $f: \domain^n \to \range^m$ satisfies generic \cref{lem-quant-generic-lemma} with $m'$ that is $\omega(\log_2 n)$. For $s > 0$, let $h(s,n) = h'(s/\log_2 K, n)$. Let $h(s,n) = h_0(s)h_1(n)$ where $h_0(1) = 1$ and $h_0$ is constant or a differentiable function where $s/h_0(s)$ is increasing and concave. Let $S^*$ be defined by:
\begin{displaymath}\frac{S^*}{h_0(S^*)} = \frac{m \  h_1(n) \  \log_2 K}{6 T}\end{displaymath}
Then 
the cumulative memory used by a quantum circuit that computes $f$ in time $T$ with error $\varepsilon\le (1-1/(2T^c))$ is at least \begin{displaymath}\frac{1}{6}\, \Loss_{h_0}(n\log_2 |\domain|)\cdot \min\left\{m\  h(S^*,n),\ 3m'\  h'(m'/2,n)\right\} \cdot \log_2 K.\end{displaymath}
Additionally if the quantum circuit uses $o(m' \log K)$ qubits, then the cumulative memory bound instead is
\begin{math}\frac{1}{6}\, \Loss_{h_0}(n\log_2 |\domain|)\cdot m\cdot h(S^*,n) \cdot \log_2 K.\end{math}
\end{corollary}
\section{Applications of our general theorems to classical and quantum computation}   
\label{sec-classical-applications}
Theorems~\ref{thm-simple-general-bound-classical} and \ref{thm:full-general-bound-classical} are powerful tools that can convert most existing time-space lower bounds into asymptotically equivalent lower bounds on the required cumulative memory.
We give a few examples to indicate how our general
theorems can be used.

\paragraph*{Unique elements}
Define $Unique_{n,N} : [N]^n \to \calP([N])$ by
$Unique_{n,N}(x)=\set{x_i\mid x_j\ne x_i\mbox{ for all }j\ne i}$.

\begin{proposition}[Lemmas 2 and 3 in \cite{DBLP:journals/siamcomp/Beame91}]
\label{prop:unique}
For the uniform distribution $\mu$ on  $[N]^n$ with $N\ge n$,
\begin{description}
\item{(A)}
$\Pr_{x \sim \mu}[\norm{Unique_{n,N}(x)} \geq n/(2e)] \geq 1/(2e-1)$
\item{(B')}
For any partial assignment $\tau$ of $k \leq n/4$ output values over $[N]$ and any restriction $\pi$ of $n/4$ coordinates in $[n]^n$,
$\Pr_{x \sim \mu}[Unique_{n,N}(x) \text{ is consistent with } \tau\mid x \text{ is consistent with } \pi] \leq e^{-k/2}$.
\end{description}
\end{proposition}

The above lemma is sufficient to prove that $TS$ is  $\Omega(n^2)$ for the unique elements problem, and can be easily extended to a cumulative complexity bound using \cref{thm:full-general-bound-classical}.

\begin{theorem}\label{thm-unique-elements}
For $n\ge N$,
any branching program
    computing $Unique_{n,N}$ in time $T$ and probability at least $4/5$ requires
    $T$ to be $\Omega(n^2/\log n)$ or
    $CM(P)$ to be $\Omega(n^2)$. 
    Further,
    any random access machine computing $Unique_{n,N}$ with probability at least $4/5$
    requires cumulative memory
    $\Omega(n^2)$
\end{theorem}

\begin{proof}
By \cref{prop:unique}, $Unique_{n,N}$ satisfies
conditions (A) and (B) of \cref{sec-generic}
with $h'(k,n)=n/4$, $m'=n/4$, $m=n/(2e)$, $C=1$,
$K=1/(2\ln N)$ and $\alpha = 1/(2e-1)\ge 0.2254$.
Since $h'(k,n)$ is independent of $k$, the
function $h_0$ defined in \cref{thm:full-general-bound-classical} is the constant function 1 and
$h_1(n)=n/4$ so $\Loss_{h_0}\equiv 1$.
We then apply \cref{thm:full-general-bound-classical} to obtain the claimed lower bounds.
\end{proof}

The above theorem is tight for $N=n$ using the
algorithm in~\cite{DBLP:journals/siamcomp/Beame91}.

\paragraph*{Linear Algebra}
We consider linear algebra over some finite field $\F$. Let $\domain$ be a subset of $\F$ with $d$ elements. 

\begin{definition}
An $m \times n$ matrix is \emph{$(g,h,c)$-rigid} iff every $k \times w$ submatrix where $k \leq g$ and $w \geq n-h$ has rank at least $ck$. We call $(g,h,1)$-rigid matrices $(g,h)$-rigid.
\end{definition}

Matrix rigidity is a robust notion of rank and is an important property for proving time-space and cumulative complexity lower bounds for linear algebra. Fortunately, Abrahamson proved that there are always rigid square matrices.

\begin{proposition}[Lemma 4.3 in \cite{Abr91}]\label{lem-rigid-matrices}
There is a constant $\gamma \in (0, \frac{1}{2})$ where at least a $1-d^{-1}(2/3)^{\gamma n}$ fraction of the matrices over $\domain^{n \times n}$ are $(\gamma n,\gamma n)$-rigid.
\end{proposition}

Abrahamson shows in \cite{Abr91} that for any constant $c \in (0,\frac{1}{2})$ and $m \times n$  matrix $A$ that is $(c m,cn,c)$-rigid, any $\domain$-way branching program that computes the function $f(x) = Ax$ with expected time $\overline{T} \geq n$ and expected space\footnote{\cite{Abr91} defines expected space as the expected value of the $\log_2$ of the largest number of a branching program node that is visited during a computation under best case node numbering.} $\overline{S}$ has $\overline{T}\overline{S} = \Omega(nm \log d)$
where $d=|\domain|$. 
We restate the key property used in that proof.

\begin{proposition}[Theorem 4.6 in \cite{Abr91}]\label{thm-Ax-probability}
Let $c \in (0,\frac{1}{2}]$, $A$ be any $m \times n$ matrix that is $(g,h,c)$-rigid and $f$ be the function $f(x) = Ax$ over $\mathbb{F}$. 
Let $\mu$ be the uniform distribution on $\domain^n$ for $\domain\subseteq \mathbb{F}$
with $|\domain|=d$. 
For any restriction $\pi$ of $h$ coordinates to values in $\domain$ and any partial assignment $\tau$ of $k \leq g$ output coordinates over $\mathbb{F}^{m}$,
\begin{displaymath}\Pr_{x\sim \mu}[f(x) \text{ is consistent with } \tau \mid x \text{ is consistent with } \pi] \leq d^{-ck}\end{displaymath}
\end{proposition}

\begin{theorem}\label{thm-mat-vec}
Let $c \in (0, \frac{1}{2}]$. Let $A$ be an $m \times n$ matrix 
over $\domain$, with $|\domain|=d$ that is $(g(m), h(n),c)$-rigid.  Then, for any $\domain$-way branching program $P$ computing $f(x) = Ax$ in $T$ steps with probability at least $n^{-O(1)}$, either $T$ is $\Omega(g(m)h(n) \log_n d)$ or $CM(P)$ is $ \Omega (g(m) h(n) \log d)$.   Further, computing $f$ on a random access machine requires cumulative memory
$\Omega(g(m) h(n) \log d)$ unconditionally.
\end{theorem}

\begin{proof}
We invoke \cref{thm-simple-general-bound-classical} using \cref{thm-Ax-probability} to obtain Property (B') with $K=d^c$ and $C=1$. Property (A) is trivial since $|f(x)| = m$.
\end{proof}
By \cref{lem-rigid-matrices} we know that for some constant $\gamma$, a random matrix has a good chance of being $(\gamma m, \gamma n)$-rigid. This means that computing $f(x) = Ax$ for a random matrix $A$ in time at most $T$ is likely to require either the cumulative memory or $T \log T$ to be $\Omega(mn \log d)$. 
Since Yesha \cite{Yes84} proved that the $n \times n$ DFT matrix is $(n/4, n/4, 1/2)$-rigid, the DFT is a concrete example where the cumulative memory or $T \log T$ is $\Omega(n^2 \log d)$; other examples
include generalized Fourier transform
matrices over finite fields~\cite[Lemma 28]{DBLP:journals/jcss/BeameJS01}.

\begin{corollary}\label{cor-fourier-mat-vec}
If $A$ is an $n \times n$ generalized Fourier transform matrix over field $\F$ with characteristic relatively prime to $n$
then any random-access machine computing $f(x) = Ax$ for $x\in \domain^n$ where $\domain\subseteq \F$ has $\norm{\domain}=d$ with probability at least $n^{-O(1)}$ requires cumulative memory that is $\Omega(n^2 \log d)$.
\end{corollary}

It is easy to see that our lower bound is asymptotically optimal in these cases.

\begin{proposition}[Theorem 7.1 in \cite{Abr91}]\label{thm-matrix-multiplication-output-bound}
Let $f:\domain^{2n^2} \to \F^{n^2}$ for 
$\domain\subseteq \mathbb{F}$ and $d=|\domain|$ be the matrix multiplication function, $\gamma$ be the constant from \cref{lem-rigid-matrices}, and $\mu$ be the uniform distribution over $(\gamma m, \gamma n)$-rigid matrices. 
Choose any integers $h$ and $k$ such that $2(h/\gamma n)^2 \leq k$. 
If $\gamma n \geq 1$ then for any $\domain$-way branching program $B$ of height $\leq h$ the probability that $B$ produces at least $k$ correct output values of $f$ is at most $d^{2-\gamma k /4}$.
\end{proposition}

\begin{theorem}\label{thm-mat-mult}
Multiplying two random matrices in $\domain^{n^2}$ with $D\subseteq \mathbb{F}$ and $d=|\domain|$ 
with probability at least $n^{-O(1)}$ requires time $T$ that is $\Omega((n^3 \sqrt{\log d})/\log n)$ or cumulative memory  $\Omega((n^6 \log d )/ T)$.
On random access machines, the cumulative memory bound is unconditional.
\end{theorem}

\begin{proof}
\cref{thm-matrix-multiplication-output-bound} lets us apply \cref{thm:full-general-bound-classical} with $m = n^2$, $h'(k,n) = \gamma n \sqrt{k/2}$, $C=d^2$, $\alpha=1$, and $K=d^{\gamma/4}$. This gives us that $h(s,n)=n \sqrt{2\gamma s/ \log_2 d}$, so $h_0(s)=\sqrt{s}$. 
Then we get that 
$\sqrt{S^*}=\frac{m n \sqrt{2\gamma /\log_2 d}\cdot  \log_2 K}{6T}$ and hence
\begin{displaymath}S^* \text{ is } \Omega \left(\frac{n^6 \log d}{T^2}\right).\end{displaymath}
Therefore we get that either
\begin{displaymath}T \text{ is } \Omega\left(\frac{n^3 \log^{1/2} d}{\log n}\right)\end{displaymath}
or, since the loss function for $h_0$ is a constant, the cumulative memory is 
\begin{displaymath}\Omega\left(\min \left( (n^6 \log d) / T,  n^5 \log^{1/2} d\right)\right).\end{displaymath}
Since the decision tree complexity of matrix multiplication is $\Omega(n^2)$, this is
$\Omega((n^6 \log d)/T)$.
For random access machines, the same cumulative memory bound applies without the condition on $T$.
\end{proof}

\subsection*{Quantum applications of the generic method}\label{sec-q-examp}
\paragraph*{Disjoint Collision Pairs Finding}
In \cite{HM21} the authors considered the problem of finding $k$ disjoint collisions in a random function $f:[m] \to [n]$, and were able to prove a time-space tradeoff that $T^3S $ is $ \Omega(k^3 n)$ for circuits that solve the problem with success probability $2/3$. Specifically, they consider circuits that must output triples $(x_{j_{2i}}, x_{j_{2i+1}}, y_{j_i})$ where $f(x_{j_{2i}}) = f(x_{j_{2i+1}}) = y_{j_i}$. To obtain this result, they prove the following theorem using the recording query technique:

\begin{proposition}[Theorem 4.6 in \cite{HM21}]\label{thm-k-collision-output-bound}
For all $1 \leq k \leq n/8$ and any quantum circuit $\calC$ with at most $t$ quantum queries to a random function $f:[m] \to [n]$, the probability that $\calC$ produces at least $k$ disjoint collisions in $f$ is at most $O(t^3 / (k^2n))^{k/2} + 2^{-k}$.
\end{proposition}
The above theorem can be extended to a lemma matching \cref{lem-quant-generic-lemma} by choosing a sufficiently small constant $\delta$ and setting $T=\delta\, k^{2/3}n^{1/3}$ to obtain a probability of at most $2^{1-k}$.
This is sufficient to obtain a matching lower bound on the cumulative memory complexity using \cref{cor-quant-generic}.
\begin{theorem}\label{thm-q-k-disj-col}
Finding $\omega(\log_2 n) \leq k \leq n/8$ disjoint collisions in a random function $f: [m] \to [n]$ with probability at least $2/3$ requires time $T$ is $\Omega(kn^{1/3}/\log n)$ or cumulative memory $\Omega(k^3 n / T^2)$.
\end{theorem}

\begin{proof}
Our discussion based on \cref{thm-k-collision-output-bound} lets us apply \cref{cor-quant-generic} with $m=m'=k, h'(k,n) = \delta k^{2/3}n^{1/3}$, and $C=K=2$. Thus we have $h(s,n) = h'(s,n)$ and $h_0$ is a differentiable function where $s/h_0(s)$ is an increasing and concave function. With these parameters, we have:
$$S^* \text{ is } \Omega\left(\frac{k^3n}{T^3}\right)$$
By \cref{cor-quant-generic} with the observation that the loss is constant we get that:
$$T \text{ is } \Omega \left(\frac{kn^{1/3}}{\log n}\right)$$
or the quantum cumulative memory is:
$$\Omega\left(\min\left(\frac{k^3n}{T^2}, k^{5/3} n^{1/3}\right)\right).$$
By \cref{thm-k-collision-output-bound} we know that any quantum circuit with at most $T' = \alpha k^{2/3}n^{1/3}$ layers can produce $k$ disjoint collisions with probability at most $2^{1-k}$. Thus we know that $T>T'$ and our cumulative memory bound becomes $\Omega(k^3n/T^2)$.
\end{proof}

\paragraph*{Linear Inequalities and Boolean Linear Algebra}

We now consider problems in Boolean linear algebra where we write $A\bullet x$ for Boolean (i.e. and-or) matrix-vector 
product and $A\bullet B$ for Boolean matrix multiplication. In \cite{KSdW07} the authors prove the following time-space tradeoff for Boolean matrix vector products:

\begin{proposition}[Theorem 23 in \cite{KSdW07}]\label{prop-bool-mat-vec}
For every $S$ in $o(n/\log n)$, there is an $n \times n$ Boolean matrix $A_S$ such that every bounded-error quantum circuit with space at most $S$ that computes Boolean matrix vector product $A_S \bullet x$ in $T$ queries requires that $T$ is  $\Omega(\sqrt{n^{3}/S})$.
\end{proposition}

This result is weaker than a standard time-space tradeoff since the function involved is not independent of the circuits that might compute
it.
In particular, \cite{KSdW07} does not find a single function that is hard for all space bounds, as the matrix $A$ that they use changes depending on the value of $S$. 
For example, a circuit using space $S' \gg S$ could
potentially compute $A_S\bullet x$ using $o(n^{3/2}/(S')^{1/2})$ 
queries.
This means that an extension of their bound to cumulative memory complexity does not follow from our \cref{cor-quant-generic}, as blocks with distinct numbers of initial qubits would be computing outputs for different functions. In \cite{ASdW09} the authors use the same space-dependent matrices to prove a result for systems of linear inequalities.

\begin{proposition}[Theorem 19 in \cite{ASdW09}]\label{prop-lin-inequalities}
Let $S$ be in $\min(O(n/t), o(n/\log n))$ and $\vec{t}$ be the all-$t$ vector. There is an $n \times n$ Boolean matrix $A_S$ such that every bounded error quantum circuit using space $S$ for evaluating the system $A_S x \geq \vec{t}$ using $T$ queries requires $T$ that is $\Omega(\sqrt{(tn^3/S)})$.
\end{proposition}

Again this result is not a general time-space tradeoff and hence is not compatible with  obtaining a true cumulative memory bound\footnote{The analogous cumulative complexity result would require the matrix $A$ to depend extensively on the structural properties of the circuit, including the number of qubits after each layer and the locations of each fixed output gate.  It is unclear whether the TS results also may need the matrix $A_S$ to depend on the locations of the output gates.}. While neither of the above results is a time-space
tradeoff for a fixed function,
\cite{KSdW07} leverages the ideas for
\cref{prop-bool-mat-vec} to compute a true
time-space tradeoff lower bound for computing 
Boolean matrix multiplication. 

\begin{proposition}[Theorem 25 in \cite{KSdW07}]
\label{prop-quantum-matrix-product}
If a quantum circuit computes the Boolean matrix product $A\bullet B$ with bounded error using $T$ queries and $S$ space, then $TS$ is $\Omega(n^5/T)$.
\end{proposition}

In \cref{prop-quantum-matrix-product}, unlike in \cref{prop-bool-mat-vec} and \cref{prop-lin-inequalities},
both $A$ and $B$ are inputs to the problem. 
This allows the lower bound argument to use the
properties of the circuit to find matrices $A$ and $B$
for which the circuit will be particularly
challenged.
More precisely, to prove the above result, the authors use a lemma matching the form of \cref{lem-quant-generic-lemma} that 
 we extract from their lower bound argument.
\begin{proposition}[from Theorem 25 in \cite{KSdW07}]\label{lem-quant-matrix-mult-lemma}
Let $R \subseteq [n] \times [n]$ be any fixed set of $k \in o(n)$ outputs to the function $f(A,B) = A\bullet B$. Then there are constants $\alpha, \gamma > 0$ such that for any quantum circuit $\calC$ with at most $\alpha \sqrt{kn}$ layers, there is a distribution $\mu_\calC$ over pairs of matrices such that when $(A,B) \sim \mu_\calC$, the probability that $\calC$ produces the correct values for $R$ is at most $2^{-\gamma k}$.
\end{proposition}

Note that, though there are $\Omega(n^2)$
total output values,
\cref{lem-quant-matrix-mult-lemma} only works when k --- the number of output values in a block --- is sublinear in $n$. 
This is not a problem in the time-space tradeoff lower bound. \cref{lem-quant-matrix-mult-lemma} upper bounds the value of $k$ for a block as $O(S)$.
Since the time $T$ must be $\Omega(n^2)$ simply to read the input, the bound $T^2S = \Omega(n^5)$ trivially holds when $S$ is $\Omega(n)$. Thus the time-space tradeoff proof only needs to apply \cref{lem-quant-matrix-mult-lemma} when $S$ (and therefore $k$) is sublinear in $n$.

We cannot apply such an argument when considering cumulative memory complexity, as a circuit can use $\Omega(n)$ qubits for a small number of layers without having an asymptotic effect on the cumulative memory complexity.
However, if we consider $o(n)$ space bounded computation, we can get a matching bound on the cumulative memory complexity.

\begin{theorem}\label{thm-q-bool-mat-mul}
Any quantum circuit that computes the Boolean matrix product $A\bullet B$ requires $\Omega(n)$ ancilla qubits or cumulative memory that is $\Omega(n^5/T)$.
\end{theorem}

\begin{proof}
\cref{lem-quant-matrix-mult-lemma} lets us apply \cref{cor-quant-generic} with $m'$ being $o(n),\ m=n^2,\ h'(k,n)=\alpha \sqrt{kn}$, and $K = 2^{1/\gamma}$. Thus we have $h(s,n) = h'(s/\gamma,n) = \alpha \sqrt{sn/\gamma}$ and $h_0(s) = \sqrt{s}$. Therefore we define $S^*$ to be
$$S^* = \frac{\gamma \alpha^2 n^5}{36 T^2}$$
Thus by \cref{cor-quant-generic} we get that the space bound is $\Omega(n)$ or the cumulative memory is $\Omega(n^5/T)$.
\end{proof}

Though this is somewhat limited in its range of applicability, it still yields
a generalization of the time-space
tradeoff lower bound of \cref{prop-quantum-matrix-product} when $S$ is $o(n)$.
\section{Cumulative memory complexity of single-output functions}
\label{sec-single}

The time-space tradeoff lower bounds known for classical algorithms computing single-output functions are quite a bit weaker than those for multi-output functions, but the bounds we can obtain on cumulative memory for slightly super-linear time bounds are nearly as strong as those for multi-output functions. 

For simplicity we focus on branching programs with Boolean output, in which case, we can simply assume that the 
output is determined by which 
of two nodes the branching
program reaches at time step $T$.

The general method for bounds for single output functions is based on the notion of the \emph{trace} of a branching program computation. 
We fix a branching program $P$ computing $f:\domain^n\rightarrow \set{0,1}$.
As in the case of the simple bounds for multi-output functions, we break up $P$ into a sequence of blocks, say $\ell$ of them, that are separated by time steps
$0=t_1,\ldots, t_{\ell}, t_{\ell+1}=T$.
A trace $\tau$ in $P$ is a sequence of $\ell$ nodes of $P$, one node in the set of nodes $L_{t_i}$ at time step $t_i$ for each $i=1,\ldots, \ell$.   
The set of all traces $\mathcal{T}=L_{t_1}\times \cdots \times L_{t_{\ell}}$.

A key object under consideration is the notion of an \emph{embedded rectangle}, which is a subset
of $R\subseteq \domain^n$ with associated disjoint subsets $A\subset [n]$ and $B\subset [n]$ with
$|A|=|B|=m(R)=m$ and assignment $\sigma\in \domain^{[n]-A-B}$ such that $R=R_A\times R_B\times \sigma$.
We write $\alpha(R)=\min(|R_A|, |R_B|)/|\domain|^m$.

\begin{proposition}[Implicit in Corollary 5.2 of~\cite{DBLP:journals/jacm/BeameSSV03}]
\label{rect-prop}
Let $P$ be a branching program  of length $T$ computing a function $f:\domain^n\rightarrow \{0,1\}$.  
Suppose that
$T\le kn$ for $k\ge 4$ and
$n\ge \ell\ge k^2 2^{k+6}$. 
If $0=t_1 < t_2 < \cdots < t_{\ell+1}=T$ are time steps with $t_{i+1}-t_i\le n/(k 2^{k+6})$, 
then there is an embedded rectangle $R\subseteq f^{-1}(1)$ with $m(R)=m\ge n/2^{k+1}$ and
$\alpha(R)\ge 2^{-12 (k+1) m-2}\cdot |\mathcal{T}|^{-1}\cdot |f^{-1}(1)|/|\domain|^n$
where $\mathcal{T}$ is the set of traces of $P$ associated with time steps $t_1,\ldots, t_\ell$.
\end{proposition} 

\begin{corollary}
\label{cor:cm-single-1}
Let $P$ be a $\domain$-way branching program of length $T$ computing a function $f:\domain^n\rightarrow \{0,1\}$. 
If 
$T\le kn$ for $k\ge 4$ and
$n\ge k^2 2^{k+8}$,
then there is an embedded rectangle $R\subseteq f^{-1}(1)$ with $m(R)=m\ge n/2^{k+1}$ and
$\alpha(R)\ge 2^{-12 (k+2) m -k\cdot 2^{k+9} \cdot CM(P)/n -2} \cdot |f^{-1}(1)|/|\domain^n|$.
\end{corollary}

\begin{proof}
Fix a branching program $P$ of length $T\le kn$ computing $f$.   
We can extend $P$ to length exactly $kn$ by adding a chain
of nodes to the root.  This does not impact the cumulative memory bound of $P$ -- a single node per level is 0 space -- so we assume that $T=kn$ without loss of generality.
Let $\ell = k^2 2^{k+8}$.
We apply the same basic idea for the choice of time steps $0=t_1,t_1,\ldots, t_{\ell+1}=T$ used in the simple general method for multi-output functions:  Namely, we break $P$ into $\ell$ time segments of length either $h=\lfloor kn/\ell\rfloor$ or $\lceil kn/\ell\rceil$. 
We define $t_1=0$ and define $t_i$ for $1<i\le \ell$ to be the time step during the next segment at which the set $|L_{t_i}|$ is minimized.   Write $S_i=\log_2 |L_{t_i}|$.
Then the cumulative memory complexity used by $P$ satisfies
\begin{displaymath}CM(P)\ge \sum_{i=1}^\ell S_i \cdot h=h\cdot \log_2 |\mathcal{T}|,\end{displaymath}
since $|\mathcal{T}|=\prod_{i=1}^t |L_{t_i}|$.

Clearly each $t_{i+1}-t_i$ is at most $2\lceil kn/\ell\rceil\le n/(k 2^{k+6})$ by definition, since their difference is at most the length of two consecutive time segments.  
Therefore, the conditions of~\cref{rect-prop} apply and we obtain that there is an embedded rectangle $R\subseteq f^{-1}(1)$
with $m(R)\ge n/2^{k+1}$ and 
\begin{align*}
\alpha(R)&\ge  2^{-12 (k+2) m -2}\cdot |\mathcal{T}|^{-1} \cdot |f^{-1}(1)|/|\domain^n|\\
&\ge 2^{-12 (k+2) m -2 - CM(P)/h } \cdot |f^{-1}(1)|/|\domain^n|\\
&\ge 2^{-12 (k+2) m  -k\cdot 2^{k+9} \cdot CM(P)/n -2} \cdot |f^{-1}(1)|/|\domain^n|.\qedhere
\end{align*}
\end{proof}

An example of a natural problem that we can apply this to is the Hamming Closeness problem
$HAM_{1/8,n,N}:[N]^n\rightarrow \{0,1\}$
which outputs 1 iff there is a pair of input coordinates $x_i, x_j\in [N]$ such that the Hamming distance
between the binary representations of $x_i$ and $x_j$ is at
most $\frac18 \log_2 N$.

\begin{proposition}[\cite{DBLP:journals/jacm/BeameSSV03}]
\label{prop-ham}
For $f(x)=1-HAM_{1/8,n,N}(x)$, and
$N\ge n^{4.39}$ we have
\begin{itemize}
\item (Proposition 6.15) $|f^{-1}(1)|\ge N^n/2$, and
\item (Lemma 6.17) there
is a constant $\beta>0$ such that any embedded rectangle $R\subseteq f^{-1}(1)$ has $\alpha(R)\le N^{-\beta m(R)}$.
\end{itemize}
\end{proposition}

We can apply the above to prove that 
any $[N]$-way branching program computing
$HAM_{1/8,n,N}$ for
$N\ge n^{4.39}$ in time $T$ and space
$S$ requires
$T$ that is $\Omega(n \log\left(\frac{n\log n}{S}\right))$.

\begin{theorem}\label{thm-ham-close}
For $N\ge n^{4.39}$ any $[N]$-way branching program computing
    $HAM_{1/8,n,N}$ in time $T$ that
    is $o(n \log n)$
    requires cumulative memory
    $(n^2 \log n)/2^{O(T/n)}$ which is
    $n^{2-o(1)}$.
\end{theorem}

\begin{proof}
Let $P$ be an $[N]$-way branching program computing
$HAM_{1/8,n,N}$ in time $T$ that
is $o(n\log n)$.
We can swap the sink nodes to obtain
a branching program $P'$ computing
$f=1-HAM_{1/8,n,N}$.
Write $k=T/n$ and assume wlog that $k\ge 4$.
Therefore $k$ is $o(\log n)$ and
hence $k^2 2^{k+8}$ is $n^{o(1)}$ and
hence $\le n$.
Therefore by \cref{cor:cm-single-1},
there is an embedded rectangle 
$R\subseteq f^{-1}(1)$ such that 
$m(R)=m\ge n/2^{k+1}$ and
\begin{displaymath}\alpha(R)\ge 2^{-12 (k+2) m -k\cdot 2^{k+9} \cdot CM(P')/n -2} \cdot |f^{-1}(1)|/N^n.\end{displaymath}
Therefore by \cref{prop-ham}, for some constant $\beta>0$ we have
\begin{displaymath}N^{-\beta m}\ge \alpha(R)\ge2^{-12 (k+2) m -k\cdot 2^{k+9} \cdot CM(P')/n -3}. \end{displaymath}
Since $CM(P)=CM(P')$, solving we obtain
\begin{displaymath}k\cdot 2^{k+9} \cdot CM(P)\ge \beta n m \log_2 N -12 (k+2)m n -3n.\end{displaymath}
Since $k+2$ is $o(\log N)$ we obtain
that $k\cdot 2^{k+9}\cdot CM(P)\ge \delta n m \log_2 N$ for some constant
$\delta>0$.
Therefore, plugging in the value of $T/n$ for $k$, we see that $CM(P)$ is  $(n^2 \log n)/2^{O(T/n)}$.  
This is $n^{2-o(1)}$ by the bound on $T$.
\end{proof}

Similar bounds can also be shown by related means for various problems involving computation of quadratic
forms, parity-check matrices of codes and others.
For some problems the following stronger lower bound
method is required.

\begin{proposition}[Implicit in Corollary 5.4 of~\cite{DBLP:journals/jacm/BeameSSV03}]
\label{rect-boolean-prop}
Let $P$ be a $\domain$-way branching program of length $T$ computing a function $f:\domain^n\rightarrow \{0,1\}$.   Suppose that
$T\le (k-2)n$ for $k\ge 8$ and
$n \ge \ell \ge 2q^{5k^2}$ for $q\ge 2^{40}k^{8}$. 
If $0=t_1 < t_2 < \cdots < t_{\ell+1}=T$ are time steps with $t_{i+1}-t_i\le kn/q^{5k^2}$, 
then there is an embedded rectangle $R\subseteq f^{-1}(1)$ with $m(R)=m\ge q^{-2k^2}
 n/2$ and
$\alpha(R)\ge 2^{-q^{-1/2} m}\cdot |\mathcal{T}|^{-1}\cdot |f^{-1}(1)|/|\domain|^n$
where $\mathcal{T}$ is the set of traces of $P$ associated with time steps $t_1,\ldots, t_\ell$.
\end{proposition} 

\begin{corollary}
\label{cor:cm-single-2}
Let $P$ be a branching program of length $T$ computing a function $f:\domain^n\rightarrow \{0,1\}$. 
If 
$T\le (k-2)n$ for $k\ge 8$ and
$n\ge  2q^{5k^2}$ for
$q=2^{40}k^{8}$,
then there is an embedded rectangle $R\subseteq f^{-1}(1)$ with $m(R)=m\ge q^{-2k^2}
 n/2$ and
$\alpha(R)\ge 2^{-q^{-1/2} m - q^{5k^2} CM(P)/n} \cdot |f^{-1}(1)|/|\domain^n|$.
\end{corollary}

\begin{proof}[Proof Sketch]
The proof is the analog of
that of \cref{cor:cm-single-1} using \cref{rect-boolean-prop} in place of \cref{rect-prop}.
\end{proof}

Define the Element Distinctness function 
$ED_{n,N}$ on $[N]^n$ to be the Boolean function
that is 1 iff all values in the input are distinct.
\begin{proposition}[\cite{DBLP:journals/jacm/BeameSSV03}]
\label{prop-ed}
For $N\ge n^2$,
\begin{itemize}
    \item (Proposition 6.11) $|ED_{n,N}^{-1}(1)|\ge N^n/e$, and
    \item (Lemma 6.12)  Every
    embedded rectangle $R$ in $ED_{n,N}^{-1}(1)$ has $\alpha(R)\le 2^{-m(R)}$.
\end{itemize}
\end{proposition}

\cite{DBLP:journals/jacm/BeameSSV03} used this to prove that the time $T$ and space $S$ for computing $ED_{n,n^2}$
must satisfy 
$T=\Omega(n\sqrt{\log(n/S)/\log\log(n/S)})$.
We strengthen this to the following
theorem using \cref{cor:cm-single-2}.

\begin{theorem}\label{thm-ed}
    Any $[n^2]$-way branching program computing
    $ED_{n,n^2}$ in time $T$ that
    is $o(n\sqrt{\log n/\log\log n})$
    requires cumulative memory
    $n^2/(T/n)^{O(T^2/n^2)}$ which is
    $n^{2-o(1)}$.
\end{theorem}

\begin{proof}
Let $P$ compute $ED_{n,n^2}$ in time $T$
that is $o(n\sqrt{\log n/\log\log n})$.
Write $k=T/n+2$ so that $T\le (k-2)/n$
and assume wlog that $k\ge 8$.
Write $q=2^{40} k^8$.
Since $T$ is
$o(n\sqrt{\log n/\log\log n})$,
$k$ is $o(\sqrt{\log n/\log\log n})$
and $2q^{5k^2}$ which is 
$k^{O(k^2)}$ and hence $n^{o(1)}$ and therefore $\le n$.
We can then apply~\cref{cor:cm-single-2}
to say that there is a rectangle $R\subseteq ED_{n,n^2}^{-1}(1)$ with
$m(R)=m\ge q^{-2k^2} n/2$ and
$\alpha(R)\ge 2^{-q^{-1/2} m - q^{5k^2} CM(P)/n} \cdot |ED_{n,n^2}^{-1}(1)|/|\domain^n|$.
By \cref{prop-ed}, we have
\begin{displaymath}2^{-m}\ge\alpha(R)\ge 2^{-q^{-1/2} m - q^{5k^2} CM(P)/n}/e.\end{displaymath}
Solving, we obtain that
\begin{displaymath}q^{5k^2} CM(P)\ge  n \cdot m (1-1/q^{1/2})-2n.\end{displaymath}
Therefore, since $m\ge q^{-2k^2}n/2$, we
have constant $c$ such
that $CM(P)\ge n^2/q^{ck^2}$.
As noted above, $q^{ck^2}$ is $n^{o(1)}$.
More precisely, the bound we obtain is
\begin{displaymath}CM(P)\ge n^2/(T/n)^{O(T^2/n^2)}.\qedhere\end{displaymath}
\end{proof}
\section{Acknowledgements}
Many thanks to David Soloveichik for his guidance and contributions to our initial results. Thanks also to Scott Aaronson for encouraging us to consider cumulative memory complexity in the context of quantum computation.
\newpage
\bibliographystyle{plainurl}
\bibliography{short-sources}
\appendix
\newpage
\section{A gap between time-space product and cumulative memory}\label{sec-cm-ts-sep}
Here we discuss some commonly studied structured sequential models of computation with provable separations between time-space and cumulative memory complexities.
\subsection*{Black pebbling separation}
\begin{definition}
The \emph{black pebble game} is a one player game played on a graph $G=(V,E)$ with source nodes $V_s \subseteq V$ and target nodes $V_t \subseteq V$. The game is played by placing and removing pebbles from the graph according to the following rules:
\begin{itemize}
    \item A pebble may be placed on any source node $v \in V_s$.
    \item A pebble may be placed on any node whose immediate predecessors all have pebbles.
    \item A pebble may be moved from a node to one of its children if all other parents of that node contain pebbles.
    \item A pebble may be removed from any node.
\end{itemize}
The goal of the game is to simultaneously have pebbles on each node $v \in V_t$. The \emph{time} of a pebbling is the number of steps taken and the \emph{pebbles} is the largest number of pebbles placed on the graph at any time. We say that the \emph{time-space} of a pebbling is the product of its time and number of pebbles. The \emph{cumulative memory} of a pebbling is the sum of the number of pebbles placed after each step.
\end{definition}

We will be considering instances of the black pebble game where $V_s$ contains exactly the unique node with in-degree zero and $V_t$ contains exactly the unique node with out-degree zero. Intuitively, the black pebble game corresponds to strategies for evaluating straight line programs, where a pebble indicates that a particular value has been computed and is currently stored in memory. Thus the number of pebbles used when pebbling a graph is analogous to the space used by that computation. We will construct a simple DAG where the time-space complexity is larger than the cumulative memory complexity.

\begin{proposition}\label{lem-graph-constr}
There is a family of DAGs $\{G_i\}_{i \in \N}$ such that graph $G_n$ requires $\Omega(n^{1/3})$ pebbles and $\Omega(n)$ steps to pebble but $G_n$ can be pebbled with cumulative memory $\Omega(n)$.
\end{proposition}

Hence there is an $\Omega(n^{1/3})$ separation between the time-space product and 
cumulative memory for pebbling.

\begin{proof}
We construct $G_n$, as shown in \cref{fig-pebbling}, to contain an $n^{1/3} \times n^{1/3}$ square lattice whose node of out-degree zero now has an out-going edge to the head of a chain containing $n$ nodes. Since pebbling the $n^{1/3} \times n^{1/3}$ lattice requires pebbling a pyramid of height $n^{1/3}$, Theorem 5 of \cite{Cook73} tells us that $n^{1/3}$ pebbles are necessary to place a pebble at the end of the lattice. Since a pebble must be placed on each node in the chain of length $n$, pebbling this graph takes at least $n$ steps.
\begin{figure}[tb]
    \centering
    \includegraphics[width=0.6\linewidth]{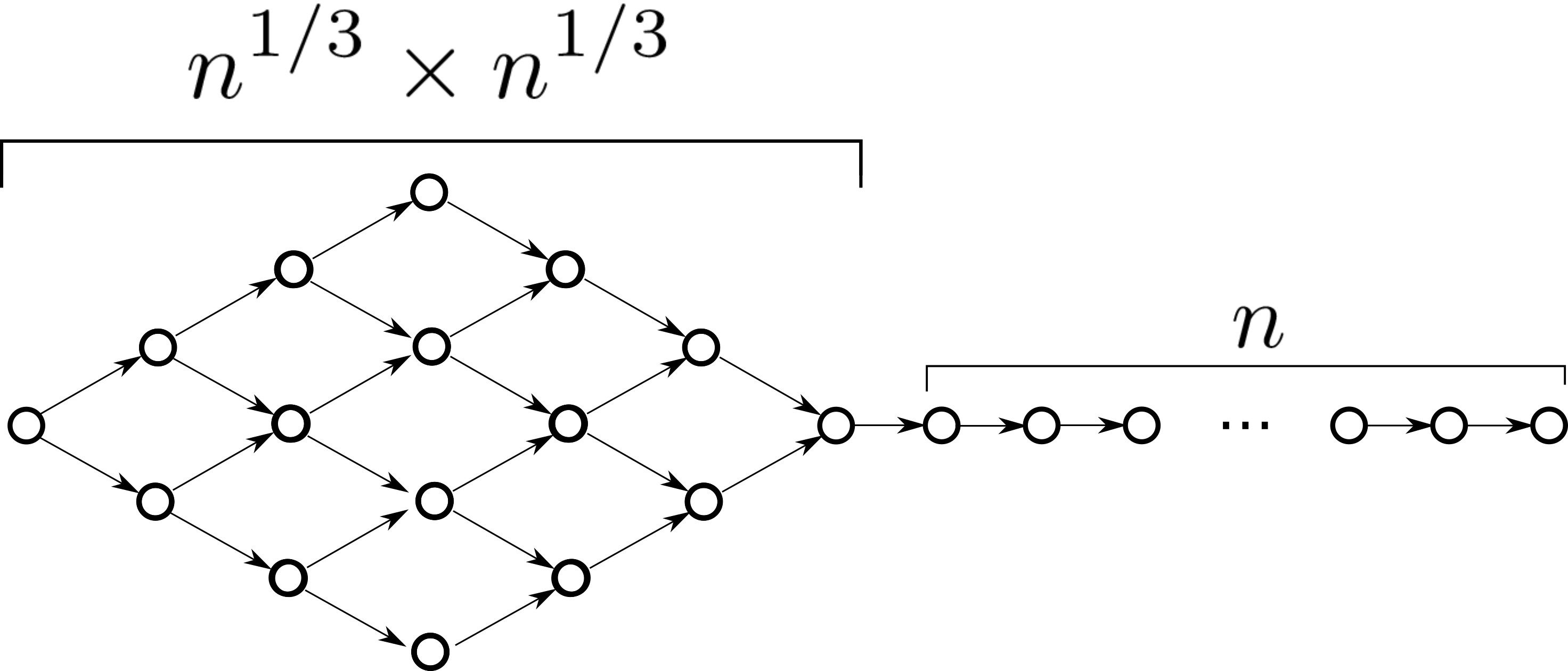}
    \caption{A DAG defined by parameter $n$. It is formed by joining an $n^{1/3} \times n^{1/3}$ lattice to a chain of length $n$.}
    \label{fig-pebbling}
\end{figure}

Now we will show how this graph can be pebbled with less cumulative memory.
We will show that both the lattice and the chain can each be pebbled with cumulative memory that is $O(n)$. The lattice can be pebbled by placing $n^{1/3}$ pebbles along the top diagonal and then repeatedly moving these pebbles along their downward edges until they are all on the bottom diagonal. This uses $n^{1/3}$ pebbles and acts on each node exactly once for a total of $n^{2/3}$ steps. Thus the cumulative memory for pebbling the lattice is $O(n)$. For the chain we can simply move one pebble from the leftmost node to the rightmost node in $n$ steps. Since this process only requires one pebble, the cumulative memory is also $O(n)$.
%
%
\end{proof}


\subsection*{Random oracle separation}
There is a group of closely related theorems from cryptography that let us instantiate pebbling graphs with the help of a random oracle \cite{DNW05, DKW11, KK14}. Here we walk through the ideas behind these proofs to show that the graph $G_n$ in \cref{lem-graph-constr} leads to separation between time-space and cumulative memory complexities in the random oracle model. The concrete problem we will be considering is related to labeling nodes of the pebbling graph.

\begin{definition}
Let $G = (V,E)$ be a DAG with maximum in-degree two and target $v_t$. Fix $c$ to be some large constant. Let $\calH: \{0,1\}^{(2c+1)\ceil{\log_2 |V|}} \to \{0,1\}^{c \ceil{\log_2 |V|}}$ be a random function. We assign each $v_i \in V$ a label $L(v_i)$ as follows:
\begin{itemize}
    \item If $v_i$ has in-degree zero, then $L(v_i) = \calH(0^{c \ceil{ \log_2 |V|}},0^{c \ceil{ \log_2 |V|}},i)$.
    \item If $v_i$ has exactly one parent $v_j$, then $L(v_i) = \calH(L(v_j), L(v_j), i)$.
    \item If $v_i$ has two parents $v_j, v_k$ where $j < k$, then $L(v_i) = \calH(L(v_j), L(v_k), i)$.
\end{itemize}
The \emph{hash-graph} problem $H_{G}^\calH$ is the task of computing the label of $v_t$.
\end{definition}

We will use the notation $\calA^{\calH}$ to denote an algorithm $\calA$ that has query access to the random oracle (function) $\calH$. We start by proving a weaker version of a result in \cite{DNW05}.

\begin{definition}[\cite{DNW05}]
Let $G = (V,E)$ be a DAG and $\calA^{\calH}$ be an algorithm that solves the hash-graph problem $H_G^{\calH}$. Then the \emph{ex post facto} pebbling of $\calA^{\calH}$ is defined as follows:
\begin{itemize}
    \item Making the call $\calH(0^{c \ceil{ \log_2 |V|}},0^{c \ceil{ \log_2 |V|}},i)$ when $v_i$ has in-degree zero corresponds to placing a pebble on $v_i$.
    \item Making the call $\calH(L(v_j), L(v_j), i)$ when $v_i$'s only parent is $v_j$ corresponds to placing a pebble on $v_i$.
    \item Making the call $\calH(L(v_j), L(v_k), i)$ when $v_j$ and $v_k$ are the parents of $v_i$ and $j<k$ corresponds to placing a pebble on $v_i$.
    \item A pebble is remove as soon as it is no longer needed. This happens when either the children of that node are never pebbled after this point or when the node is pebbled again before any of its children are pebbled.
\end{itemize}
\end{definition}

Analyzing ex post facto pebbling is key to the arguments in \cite{DNW05, DKW11, KK14} and lets us lower bound the space required to compute a hash-graph.

\begin{proposition}[\cite{DNW05}]\label{prop-pebbles-bounded-by-space}
Consider an algorithm $\calA^\calH$ that operates for a certain number of steps with $s \cdot c \ceil{\log_2 |V|}$ bits of memory. Then with probability at least $1-1/|V|^c$ (over the choice of $\calH$) the maximum number of pebbles placed by the ex post facto pebbling of $\calA^\calH$ is bounded above by $s$.
\end{proposition}

This lets us show that a hash-graph problem requires at least as much time and space as pebbling its underlying graph, as we show below in an argument similar to ones in \cite{DKW11, KK14}.

\begin{proposition}\label{lem-pebbling-is-space-bound}
Let $G=(V,E)$ be a DAG that requires $s$ pebbles and $\tau$ steps to pebble. Then any algorithm $\calA^\calH$ that solves the hash-graph problem $H_G^{\calH}$ must use space $S$ that is larger than $(s-1)\cdot c \ceil{\log_2 |V|}$ and time $T$ that is at least $\tau$ or its success probability (over the randomness of $\calH$) is at most $2T/|V|^{c-1}$.
\end{proposition}
\begin{proof}
For an algorithm $\calA^{\calH}$ with space bound $S < (s-1)\cdot c \ceil{\log_2 |V|}$ to solve $H_{G}^{\calH}$, one of the following events must happen:
\begin{enumerate}
    \item $\calA^{\calH}$ solves $H_{G}^{\calH}$ without placing a pebble on the target during the ex post facto pebbling.
    \item $\calA^{\calH}$ places $s$ pebbles during the ex post facto pebbling.
    \item $\calA^{\calH}$ places a pebble during the ex post facto pebbling that would not be valid according to the black pebble game.
\end{enumerate}
Since $\calH$ is a random oracle, (1) happens with probability at most $1/|V|^c$. By \cref{prop-pebbles-bounded-by-space} (2) happens with probability at most $1/|V|^c$. Since guessing a label that has not been pebbled is possible with probability at most $1/|V|^c$, We know that (3) happens with probability at most $T/|V|^{c-1}$ via a union bound over the queries of $\calA^{\calH}$ and the nodes of $G$. Thus by a union bound, the probability $\calA^{\calH}$ can produce the correct output is at most $(T \cdot |V| + 2)/|V|^c$ which in turn is at most $2T/|V|^{c-1}$.

Now consider an algorithm $\calA^{\calH}$ with time bound $T < \tau$. Since $G$ cannot be pebbled in this number of steps, one of the following events must happen for $\calA^{\calH}$ to produce the correct output:
\begin{enumerate}
    \item $\calA^{\calH}$ solves $H_G^{\calH}$ without placing a pebble on the target during the ex post faco pebbling.
    \item $\calA^{\calH}$ places a pebble during the ex post facto pebbling that would not be valid according to the black pebble game.
\end{enumerate}
Both of the events are the same as in the space bounded case, and therefore a union bound give $\calA^{\calH}$ a success probability of at most $2T/|V|^{c-1}$.
\end{proof}
Thus any algorithm that solves $H_G^{\calH}$ must obey the space and the time bounds imposed by pebbling that graph or spend time that is $\Omega(|V|^{c-1})$.\footnote{Since $c$ is an arbitrary constant, this time bound can be made arbitrarly large.} Note that a pebbling of a graph $G$ directly corresponds to a strategy for computing $H_G^{\calH}$ with the same space, time, and cumulative memory bounds as the pebbling. By using the graph $G_n$ from \cref{lem-graph-constr}, this gives us a separation between time-space product complexity and cumulative memory in the random oracle model.

\begin{corollary}
Relative to a random oracle $\calH$, there is a problem with an $\Omega(n^{1/3})$ separation between its time-space product complexity and its cumulative memory complexity.
\end{corollary}

While this will be hard to prove, we believe that replacing the random oracle with a suitable hash function gives a problem where there is an asymptotic gap between these complexity measures.

\begin{conjecture}
Instantiating the family of DAGs from \cref{lem-graph-constr} as hash-graph problems with some concrete hash-function gives a problem with an unconditional asymptotic gap between its sequential time-space product and cumulative memory complexities.
\end{conjecture}
\newpage
\section{Quantum sorting with arbitrary success probability}\label{sec-quant-sort-output-bounds}

In order to modify \cref{thm-quantum-sort-lb} for different success probabilities, the key is to choose a different value for the parameter $\gamma$ in \cref{thm-q-kthresh-dp} governing the completeness bound, which will change the corresponding value of $\alpha$ and $\beta$. If we want to deal with non-constant values of $\gamma$, it is important to understand how $\gamma$ and $\alpha$ are related in \cref{thm-q-kthresh-dp}.
The following lemma is sufficient to prove Theorem 13 in \cite{KSdW07} (our \cref{thm-q-kthresh-dp}). Although the authors of~\cite{KSdW07} prove a more general version of this proposition,
the statement below captures what is necessary in our proof. Specifically, we invoke their Lemma 12 where $\delta = 0, C= ke^{\gamma +1}$ and $D =\alpha \sqrt{kn}$.

\begin{proposition}[Special case of Lemma 12 in \cite{KSdW07}]\label{lem-polynomial-threshold}
Let $p$ be a degree $2\alpha \sqrt{kn}$ univariate polynomial such that:
\begin{itemize}
    \item $p(i) = 0$ when $i \in \{0, \ldots, k-1\}$
    \item $p(k) = \sigma$
    \item $p(i) \in [0,1]$ when $i \in \{k+1, \ldots, n\}$
\end{itemize}
Then there exists universal positive constants $a$ and $b$ such that for any $\gamma > 0$ where $ke^{\gamma + 1} \leq n-k$:
\begin{displaymath}\sigma \leq a \cdot \text{exp}\left(\frac{b(2 \alpha \sqrt{kn} - k)^2 + 4e^{\gamma/2 + 1/2} k \sqrt{n-k}(2\alpha\sqrt{n} - \sqrt{k})}{n - k(e^{\gamma +1} + 1 )} - k - \gamma k \right).\end{displaymath}
\end{proposition}

The $\sigma$ in this bound gives the completeness bound on the $k$-threshold problem. We now prove that $\sigma$ is sufficiently small when $\alpha \in \Omega(e^{-\gamma/2})$.

\begin{sloppypar}
\begin{lemma}\label{lem-k-thresh-consts}
Let $a,b>0$ be the constants from \cref{lem-polynomial-threshold}.
When we have $\sqrt{k/n} < \alpha < \min\left(1/(16\sqrt{e^{\gamma + 1} + 1}), \sqrt{1 / (8b)}\right)$, the completeness bound $\sigma$ for the $k$-threshold problem with $\alpha \sqrt{kn}$ queries is less than $a\cdot e^{-\gamma k}$.
\end{lemma}
\end{sloppypar}

\begin{proof}
By \cref{lem-polynomial-threshold} we have
\begin{align}
    \sigma & \leq a \cdot \text{exp}\left(\frac{b(2 \alpha \sqrt{kn} - k)^2 + 4e^{\gamma/2 + 1/2} k \sqrt{n-k}(2\alpha\sqrt{n} - \sqrt{k})}{n - k(e^{\gamma +1} + 1 )} - k - \gamma k \right)
    \label{eq12}
\end{align}
We first bound the first term in the numerator
\begin{align*}
    b(2 \alpha \sqrt{kn} - k)^2&=b(4\alpha^2 kn - 4 \alpha k \sqrt{kn} + k^2)\\
    &=kb\alpha^2(4n - 4\sqrt{kn}/\alpha+k/\alpha^2)\\
    &< kb\alpha^2(4n - 3k/\alpha^2)
    \qquad\textrm{since $\sqrt{k/n} < \alpha$}\\
    &< kb\alpha^2(4n - 4k(e^{\gamma +1} + 1 ))\qquad\textrm{since  $\alpha < 1/(16 \sqrt{e^{\gamma + 1} +1})$}\\
    &= 4kb\alpha^2 (n-k(e^{\gamma +1} + 1 )).
\end{align*}
Next we bound the second term in the numerator
\begin{align*}
    4e^{\gamma/2 + 1/2} k \sqrt{n-k}(2\alpha \sqrt{n} - \sqrt{k}) &\leq 4e^{\gamma/2 + 1/2} k \sqrt{n} (2\alpha \sqrt{n} - \sqrt{k})\\
    &=4e^{\gamma/2 + 1/2} k \alpha (2 n - \sqrt{nk}/\alpha)\\
    &< 4e^{\gamma/2 + 1/2}k \alpha(2n-k/\alpha^2)\qquad \textrm{since $\sqrt{k/n} < \alpha$}\\
    &< 4e^{\gamma/2 + 1/2}k \alpha(2n-2k(e^{\gamma+1}+1)) \qquad \textrm{since $\alpha < 1/(16 \sqrt{e^{\gamma+1} +1})$}\\
    &=8e^{\gamma/2 + 1/2}k \alpha(n-k(e^{\gamma+1}+1))
\end{align*}

Plugging these bounds into (\ref{eq12}) we get
\begin{align*}
    \sigma&< a \cdot \text{exp}\left(4kb\alpha^2 + 8e^{\gamma/2+1/2}k \alpha - k - \gamma k\right)\\
    &< a \cdot \text{exp}\left(4kb\alpha^2 - k/2 - \gamma k\right) \qquad \textrm{since $\alpha < 1/(16 \sqrt{e^{\gamma+1} +1})$}\\
    &< a \cdot \text{exp}\left(-\gamma k\right) \qquad \textrm{since $\alpha < \sqrt{1 /8b}$} \qedhere
\end{align*}
\end{proof}

We now describe how to substitute the bound of \cref{lem-k-thresh-consts} in place of \cref{thm-q-kthresh-dp} to yield cumulative memory lower bounds for quantum circuits with failure probability at most $\delta$:
To prove the analog of \cref{cor-quantum-sort-outputs}, for any $S$, $k$, and $\delta \in (0,1)$, we can reprove a more precise version of \cref{lem-qcircuit-fixed-outputs-bound} using \cref{lem-k-thresh-consts} instead of \cref{thm-q-kthresh-dp} to bound the success probability for the $k$-threshold problem and choose
\begin{displaymath}\gamma =\frac{\ln(a \cdot 2^{2S}/(1-\delta)) - 1}{k} +1
\end{displaymath}
to get a success probability of less than $1-\delta$ for circuits with as many as
\begin{displaymath} \Omega\left(\left(\frac{1-\delta}{ 2^{2S}}\right)^{1/2k}\sqrt{kn}\right)\end{displaymath} 
layers and $S$ qubits of advice to produce $k$ outputs. If we repeat the proof of \cref{thm-quantum-sort-lb} for failure probability less than $\delta$, we can set $\beta$ to a value that is $\Omega(\sqrt{1-\delta})$ to obtain a lower bound on the cumulative memory that is $\Omega((1-\delta)n^3/T)$.
\newpage
\section{Optimizations}\label{append-optimization-lemma}
In this section we prove general optimization lemmas that allow us to derive worst-case properties of the 
allocation of branching program layers into blocks.

The first special case is 
relevant for our analysis of 
quantum sorting algorithms.
\begin{lemma}\label{lem-moments}
For non-negative reals $x_1, x_2,\ldots$ if
$\sum_i x_i \le \sum_i x_i^2$  then $\sum_i x_i^3 \ge \sum_i x_i^2$.
\end{lemma}
\begin{proof}
Without loss generality we remove all
$x_i$ that are 0 or 1 since they 
contribute the same amount to each of
$\sum_i x_i$, $\sum_i x_i^2$, and $\sum_i x_i^3$.
Therefore every $x_i$ satisfies $0<x_i<1$
or it satisfies $x_i>1$.
We rename those $x_i$ with
$0<x_i<1$ by $y_i$ and those $x_i$ with
$x_i>1$ by $z_j$.

Then $\sum_i x_i \le \sum_i x_i^2$ can  be rewritten as $\sum_i y_i (1- y_i) \le\sum_j z_j (z_j-1)$,
and both quantities are positive.
Let $y^*$ be the largest value $<1$ and
$z^*$ be the smallest value $>1$.
Thus:
\begin{align*}
    \sum_i (y_i^2 - y_i^3)
    &=\sum_i y_i^2 (1-y_i)
    \le \sum_i y^* y_i (1-y_i)
    = y^* \sum_i y_i (1-y_i)
    \le y^* \sum_j z_j (z_j-1)\\
    &< z^* \sum_j z_j (z_j-1)
    = \sum_j z^* z_j (z_j-1)
    \le \sum_j z_j^2 (z_j-1)
    =\sum_j (z_j^3 - z_j^2).
\end{align*}
Rewriting gives $\sum_i y_i^2 +\sum_j z_j^2 < \sum_i y_i^3 +\sum_j z_j^3,$
or $\sum_i x_i^3 >\sum_i x_i^2$, as required.
\end{proof}

The following is a generalization of the above to all differentiable functions $p:\mathbb{R}^{\ge 0}\rightarrow \mathbb{R}^{\ge 0}$ such that $s/p(s)$ is a concave increasing function of $s$.

\concavelemma*

\begin{proof}
By hypothesis,
\begin{math}
\sum_i \left(x_i - K p(x_i)/L\right)\ge 0
\end{math}.
Observe that $s-Kp(s)/L$
is an increasing function of $s$ since $s/p(s)$ is an increasing function
of $s$ that is 0
precisely when $s = q^{-1}(K/L)$.  
Since all $x_i$ with $x_i=q^{-1}(K/L)$ evaluate to 0 in the sum, we can rewrite it as
\begin{equation}
    \sum_{x_i> q^{-1}(K/L)} \left(x_i - K p(x_i)/L\right)\ge \sum_{x_i<q^{-1}(K/L)} \left(K p(x_i)/L-x_i\right),\label{eqn-renamed}
\end{equation}
where each of the summed terms is positive.
For $x_i\ne q^{-1}(K/L)$, define \begin{displaymath}f(x_i)=x_i \cdot \frac{p(x_i) - q^{-1}(K/L)\cdot L/K}{x_i - K p(x_i)/L}.\end{displaymath}
Observe that for $x_i =q^{-1}(K/L)$ 
the denominator is 0 and 
the numerator equals $p(x_i)-x_i \cdot L/K$ which is also 0.
For $x_i>q^{-1}(K/L)$
both the numerator and denominator are
positive and for $x_i<q^{-1}(K/L)$ both
the numerator and denominator are negative.   
Hence $f(x_i)$ is non-negative for every $x_i\ne q^{-1}(K/L)$.

\begin{claim}  
If $q$ is a convex differentiable function, we can complete $f$ to a (non-decreasing) continuous
function of $x$  
with $f'(x)\ge 0$ for all $x$ with $0<x\ne q^{-1}(K/L)$
\end{claim}

\begin{proof}[Proof of Claim]
Write $a=q^{-1}(K/L)$.
Then since $p(x)>0$ and $q(a)>0$, we have
\begin{align*}
f(x)&= \frac{x\cdot p(x) - x\cdot a/q(a)}{x - q(a)\cdot p(x)}=\frac{x - (x/p(x))\cdot a/q(a)}{x/p(x) - q(a)}\\
&=\frac{x - q(x)\cdot a/q(a)}{q(x) - q(a)}
=\frac{1}{q(a)}\cdot \frac{q(a)\cdot x - a\cdot q(x)}{q(x)-q(a)}.
\end{align*}
Therefore 
\begin{align*}f'(x)&=\frac{1}{q(a)}\cdot \frac{(q(a) -a\cdot q'(x))(q(x)-q(a))-(q(a)\cdot x-a\cdot q(x))\cdot q'(x)}
{(q(x)-q(a))^2}\\
&=\frac{ q(x)-q(a) +(a-x)\cdot q'(x)}{(q(x)-q(a))^2}.
\end{align*}
Since the denominator is a square and $q$ is increasing, to prove that $f'(x)\ge 0$ for $x\ne a$ it suffices to prove that the
numerator is non-negative.

Suppose first that $x<a$, 
Then $a-x>0$ and the numerator
$q(x)-q(a) +(a-x)\cdot q'(x)\ge 0$ if and
only if $q'(x)\ge \frac{q(a)-q(x)}{a-x}$,
which is equivalent to the slope of the tangent to $q$
at $x$ being at least that of the chord from $x$ to $a$.
This is certainly true since $q$ is a concave
function.

Suppose now that $x>a$.
Then $a-x<0$ and the numerator $q(x)-q(a) +(a-x)\cdot q'(x)\ge 0$ if and only if 
$q'(x)\le \frac{q(x)-q(a)}{x-a}$.  
Again, this is true since $q$ is a concave
function.

It remains to show that we can complete $f$
to a continuous function by giving it a finite value at $a=q^{-1}(K/L)$.
By l'H\^{o}pital's rule, the limit of $q(a)\cdot f(x)$ as 
$x$ approaches $a$ is
\begin{displaymath}\frac{q(a)-a\cdot q'(a)}{q'(a)}\end{displaymath} if
the denominator is non-zero, which it is,
since $q$ is an increasing differentiable function at $a$.   
\end{proof}

We now have the tools we need.
Let $x^*_{-}$ be the largest $x_i < q^{-1}(K/L)$ and $x^*_{+}$ be the smallest
$x_i> q^{-1}(K/L)$.
Then we have $f(x^*_{+})\ge f(x^*_{-})$ and
\begin{align*}
    \sum_{x_i > q^{-1}(K/L)}& \left(x_i\ p(x_i)- q^{-1}(K/L) \cdot L/K\cdot x_i\right)\\
    &=\sum_{x_i > q^{-1}(K/L)}
    f(x_i)\cdot \left(x_i-Kp(x_i)/L\right) \\
    &\ge\sum_{x_i > q^{-1}(K/L)}
    f(x^*_{+})\cdot \left(x_i-Kp(x_i)/L\right) \\
    &\ge f(x^*_{-})\sum_{x_i > q^{-1}(K/L)}  \left(x_i-Kp(x_i)/L\right) \\
    &\ge f(x^*_{-}) \sum_{x_i < q^{-1}(K/L)}\left(Kp(x_i)/L-x_i\right) \qquad\mbox{by \cref{eqn-renamed}}\\
    &\ge \sum_{x_i < q^{-1}(K/L)}f(x_i)\cdot \left(Kp(x_i)/L-x_i\right) \\
    &=\sum_{x_i < q^{-1}(K/L)} \left( q^{-1}(K/L) \cdot L/K\cdot x_i- x_i\  p(x_i)\right).
\end{align*}
Adding back the terms 
where $x_i=q^{-1}(K/L)$, which have value 0, and rewriting we obtain
\begin{displaymath}\sum_i \left(x_i\ p(x_i) - q^{-1}(K/L) \cdot L/K\cdot x_i\right)\ge 0.\end{displaymath}
Therefore we have
\begin{displaymath}\sum_i x_i\ p(x_i)\ge q^{-1}(K/L)\cdot L/K\cdot \sum_i x_i\ge q^{-1}(K/L) \cdot (L/K)\cdot K = q^{-1}(K/L)\cdot L.\qedhere\end{displaymath} 
\end{proof}

\newpage
\section{Proof of \texorpdfstring{\cref{lem-loss}}{Lemma 5.4}}\label{sec-loss}

\losslemma*
\smallskip
\begin{proof}
Since $p$ is non-decreasing,
$1\le p^{-1}(j)\le p^{-1}(k)$ for $1\le j\le k$
and hence
\begin{equation}\frac{1}{k}\le \frac{\sum_{j=1}^k p^{-1}(j)}{k\cdot p^{-1}(k)}\le 1\label{eqn-loss1}
\end{equation}
since $p^{-1}(k)$ is included in the numerator.
$\Loss_p(n)$ is the minimum over all integers
$k\in [1,p(n)]$ of $\frac{\sum_{j=1}^k p^{-1}(j)}{k\cdot p^{-1}(k)}$ and $p$ is non-decreasing so we have
$1/p(n)\le \Loss_p(n)\le 1$, which proves part (a)

When $p(s)= s^{1/c}$ we have 
\begin{displaymath}\sum_{j=1}^k p^{-1}(j)\ge \sum_{j=\lceil (k+1)/2\rceil}^k j^{\,c}> \lceil k/2\rceil (k/2)^c\ge (k/2)^{c+1}=k\cdot p^{-1}(k)/2^{c+1}\end{displaymath} so each term in the definition
of $\Loss_p(n)$ is larger than $1/2^{c+1}$ which
proves part (b).  (More precise bounds can be shown but we are not focused on the specific constant.)

Let $1\le k\le p(n)$ be an integer.
Then $1\le s=p^{-1}(k)\le n$.
Observe that there are at least $p(s)-p(s/c)$ integers $j\le k$ with $p^{-1}(j)\ge s/c$.
Therefore 
\begin{equation}\frac{\sum_{j=1}^k p^{-1}(j)}{k\cdot p^{-1}(k)}\ge \frac{(p(s)-p(s/c))\cdot s/c}{kp^{-1}(k)}=\frac{p(s)-p(s/c)} {ck} =\frac{p(s)-p(s/c)}{cp(s)}.\label{eqn-loss2}
\end{equation}
The minimum over all $k\in [1,p(n)]$ is equivalent to the minimum over all $s\in [1,n]$, which proves part (c).

Now suppose that $p$ is nice. Since $p$ is differentiable, for any $s$,
\begin{align*}
p(cs)-p(s)&=\int_{s}^{cs} p'(y)\,dy \\
   &=\int_{s/c}^{c} p'(cx) c \, dx\quad\mbox{by substitution $y=cx$}\\
   &\ge \int_{s/c}^c p'(x) \, dx\quad\mbox{since $p$ is nice}\\
   &=p(s)-p(s/c).
\end{align*}
Then by induction we have that for every positive integer $i\le \log_c s$,
$p(s)-p(s/c)\ge p(s/c^{i-1})-p(s/c^{i})$. 
Write $\ell=\lfloor\log_c s\rfloor$.
Then $s/c^\ell<c$ and
\begin{displaymath}p(s)-p(s/c^\ell)=\sum_{i=1}^{\ell} [p(s/c^{i-1})-p(s/c^{i})] \le \ell \cdot [p(s)-p(s/c)],\end{displaymath}
or equivalently that $p(s)-p(s/c)\ge (p(s)-p(s/c^\ell)/\ell$ and
hence \begin{displaymath}p(s)-p(s/c)\ge (p(s)-p(c))/\log_c s\end{displaymath}
since $p$ is a non-decreasing function.
Applying the lower bound from \cref{eqn-loss1} when $k=p(s)<2p(c)$
and the lower bound from \cref{eqn-loss2} when $p(s)\ge 2p(c)$ we obtain
\begin{displaymath}\Loss_p(n)
\ge \min\left(\frac{1}{2p(c)},\min_{1\le s\le n: p(s)\ge 2p(c)}(1-p(c)/p(s))/(c\log_c s)\right).\end{displaymath}
Since $c$ is a constant, we obtain that
$\Loss_p(n)$ is $\Omega(1/\log n)$.

Observe that $p$ given by $p(s)=1+\log_2 s$ is nice for every constant $c>0$ since
$p'(cx)=(\ln 2)^{-1}/(cx)=p'(x)/c$.
In this case we have
$p^{-1}(j)=2^{j-1}$ and $\Loss_p(n)< 2/p(n)< 2/\log_2 n$ 
since the largest term $p^{-1}(k)$ in each numerator is
(a little) more
than the sum of all smaller terms put together.
Together with the lower bound, this proves part (d).
\end{proof}
\end{document}